\documentclass{article}
\usepackage{graphicx} 

\usepackage[english]{babel}
\usepackage[utf8]{inputenc}
\usepackage[T1]{fontenc}
\usepackage[letterpaper, margin=1in]{geometry}
\usepackage{graphicx}
\usepackage{framed}
\usepackage[framemethod=tikz]{mdframed}
\usepackage{xcolor}
\usepackage[most]{tcolorbox}
\usepackage{todonotes}
\usepackage{color}

\definecolor{darkgreen}{rgb}{0,0.5,0}
\definecolor{darkgray}{rgb}{0.2,0.2,0.2}
\usepackage{hyperref}
\hypersetup{
    unicode=false,          
    colorlinks=true,        
    linkcolor=blue,         
    citecolor=purple,       
    filecolor=magenta,      
    urlcolor=cyan           
}

\usepackage{amsthm}
\usepackage{amsmath}
\usepackage{amssymb}
\usepackage{amsfonts}
\usepackage{mathrsfs}
\usepackage{mathtools}
\usepackage{verbatim}
\usepackage{footnote}

\usepackage{algorithm}
\usepackage{algorithmicx}
\usepackage[noend]{algpseudocode}

\usepackage{lineno}
\usepackage{caption}
\usepackage{framed}
\usepackage{enumerate}

\usepackage{tikzsymbols}
\usepackage{thmtools,thm-restate}
\usepackage{nicefrac}

\usepackage{subcaption}
\usepackage{makecell}
\usepackage{pdfpages}
\usepackage{multirow}
\usepackage{tabularx}
\usepackage{float}
\usepackage{booktabs}

\usepackage[capitalize, nameinlink]{cleveref}
\usepackage[section]{placeins}
\crefname{theorem}{Theorem}{Theorems}
\Crefname{lemma}{Lemma}{Lemmas}
\Crefname{invariant}{Invariant}{Invariants}
\Crefname{claim}{Claim}{Claims}
\Crefname{observation}{Observation}{Observations}
\Crefname{@algorithm}{Algorithm}{Algorithms}
\Crefname{figure}{Figure}{Figures}
\crefname{appendix}{Appendix}{Appendices}
\Crefname{appendix}{Appendix}{Appendices}

\newtheorem{theorem}{Theorem}[section]
\newtheorem{lemma}[theorem]{Lemma}

\newtheorem*{remark*}{Remark}

\newcounter{resultctr}
\renewcommand{\theresultctr}{\arabic{resultctr}}

\newtcolorbox{resultbox}[1][]{%
  enhanced,
  colback=gray!10,
  colframe=gray!60,
  coltext=black,
  boxrule=0.6pt,
  arc=2mm,
  left=1.2mm,
  right=1.2mm,
  top=1.0mm,
  bottom=1.0mm,
  before title={\refstepcounter{resultctr}}, 
  title={Result~\theresultctr\if\relax\detokenize{#1}\relax\else:~(#1)\fi},
  fonttitle=\bfseries,
  coltitle=black,
  attach title to upper,
  after title=\par\smallskip,
}

\DeclareMathOperator{\poly}{poly}

\def\LOCAL{\ensuremath{\mathsf{LOCAL}}\xspace}
\def\CONGEST{\ensuremath{\mathsf{CONGEST}}\xspace}

\newcommand{\prob}[1]{\Pr \left[ #1 \right]}
\newcommand{\eqdef}{\stackrel{\text{\tiny\rm def}}{=}}
\newcommand{\eqd}{\stackrel{\mathclap{\tiny\mbox{def}}}{=}}

\newcommand{\tO}{\tilde{O}}
\newcommand{\rb}[1]{\left( #1 \right)}
\newcommand{\cA}{\mathcal{A}}

\newcommand{\eps}{\epsilon}

\newcommand{\Sample}{\textsc{Sample}\xspace}
\newcommand{\PeelSample}{\textsc{Peel-Sample}\xspace}
\newcommand{\Dual}{\textsc{Dual}\xspace}
\newcommand{\E}[1]{\mathbb{E}\left[#1\right]}
\newcommand{\Astream}{\mathcal{A}_{\textsc{stream}}}
\newcommand{\AMPC}{\mathcal{A}_{\textsc{MPC}}}
\newcommand{\Asub}{\mathcal{A}_{\textsc{sub}}}
\newcommand{\condition}{\textsc{edge-threshold}\xspace}

\title{Sparsification Framework for Directed Densest Subgraph}

\author{Slobodan Mitrović\thanks{The authors are supported by NSF Faculty Early Career Development Program No.~2340048. e-mail: \texttt{\{smitrovic, thjpan\}@ucdavis.edu}} \\ UC Davis
\and Theodore Pan\footnotemark[1] \\ UC Davis}

\date{}

\begin{document}

\maketitle

\begin{abstract}
We develop a new approach for computing approximate directed densest subgraphs (DDS). 
Our main result is a sparsification procedure that reduces a directed graph $G$ on $n$ vertices to a graph with $n \cdot \text{poly} \log n$ edges while preserving enough structure to recover an approximate DDS of $G$.
Instantiating this framework in several memory-constrained settings, we obtain the following improvements over the state of the art:
\begin{itemize}
    \item In semi-streaming, we obtain a single-pass algorithm that computes a $(1-\varepsilon)$-approximate DDS. 
    Previously, the only semi-streaming algorithm that computed a constant approximation of DDS was by Bahmani, Kumar, and Vassilvitskii (2012), providing a $0.5-\varepsilon$ approximation in $O(\log n)$ passes.
    Hence, our work completely closes the approximation gap between undirected and directed DS in the semi-streaming setting, matching the $(1-\varepsilon)$-approximate undirected DS algorithm by Esfandiari, Hajiaghayi, and Woodruff (2016).
    \item In the near-linear-memory MPC regime, we obtain an $O(1)$-round algorithm for $(1-\varepsilon)$-approximate DDS, improving over the $O(\sqrt{\log n})$-round $(0.5-\varepsilon)$-approximation algorithm of Mitrovi\'c and Pan (2024).
    \item In the sublinear-time setting, we obtain an algorithm using $\tilde{O}(n)$ time, space, and oracle queries to compute a $(1-\varepsilon)$-approximate DDS, improving over the $\tilde{O}(n^{1.5})$ time, space, and query algorithm of Esfandiari, Hajiaghayi, and Woodruff (2016).
\end{itemize}
Prior approaches to DDS rely primarily on either uniform edge sampling or peeling. However, these approaches face inherent limitations: obtaining a constant-factor approximation via uniform sampling requires $\Omega(n^{1.5})$ sampled edges, while peeling-based methods require $\Theta(\log n)$ iterations.
Our key insight is that although many peeling iterations may be needed to compute an approximate DDS directly, the first peeling step already removes a large fraction of edges that are irrelevant to near-optimal solutions.
Building on this observation, we develop a sparsification framework that first performs a single peeling step and then uniformly samples $\tilde{O}(n)$ edges from the remaining graph, while preserving a $(1-\varepsilon)$-approximate DDS.
We expect this framework to be useful more broadly for dense-subgraph problems and related optimization tasks.
\end{abstract}

\newpage

\tableofcontents

\newpage

\section{Introduction}
Let $G = (V, E)$ be a directed graph. 
A pair of vertex subsets $(S, T) \subseteq V \times V$ is called a \emph{directed densest subgraph} -- or DDS for short -- if $|E(S, T)| / \sqrt{|S| \cdot |T|}$ is maximized over all $(S, T)$ pairs; $E(S, T)$ denotes the subset of arcs of $E$ from $S$ to $T$.
This notion of subgraph density was introduced in 1999 by Kannan and Vinay \cite{kannan1999analyzing}, motivated by its ability to capture local sub-structures in web graphs and its spectral properties ~\cite{andersen2010local}.
The DDS problem has a wide range of applications such as fraud detection \cite{prakash2010eigenspokes, hooi2016fraudar}, community and data mining \cite{kleinberg1999authoritative, ma2021directed, chen2023densest}, and computational biology \cite{junker2008analysis, saha2010dense}.

Over the years, several influential approaches have been developed to solve (approximate) undirected and directed DS in various settings, including semi-streaming, parallel, distributed, dynamic, and differentially private.

\paragraph{LP-based approaches.}
In 2000, Charikar~\cite{charikar2000greedy} formulated LPs for the undirected and directed DS problems, introducing polynomial time solutions to the DS problem.
Faster -- near-linear time -- solutions are known when approximation is allowed.
In particular, \cite{bahmani2014efficient} showed how to approximately solve the dual LP for both undirected and directed DS using the multiplicative weights update (MWU) method. Then, they construct an approximate DS from the dual solution.
Their algorithm was the first to compute a $(1-\eps)$-approximate DS while only using near-linear time in the graph size.

\paragraph{Peeling-based approaches.}
Even though the LP-based approaches yield a near-linear time $(1-\eps)$-approximate DDS, they sometimes become inefficient when the ``working'' memory is sublinear in the graph size, such as in semi-streaming.
For such settings, state-of-the-art can be achieved by \emph{peeling-based} approaches, whose original version was introduced in 1996 by Asahiro, Iwama, Tamaki, and Tokuyama~\cite{10.1007/3-540-61422-2_127} for the densest $k$-subgraph problem and a similar method was presented by Charikar~\cite{charikar2000greedy} for the densest subgraph problem.
The method constructs a sequence of graphs $G_0 = G, G_1, \ldots, G_t = \emptyset$, where $G_{i}$ is obtained by removing -- also called peeling -- the \emph{smallest-degree} vertices from $G_{i-1}$.
It can be shown that at least one among these graphs is a $0.5$-approximate densest subgraph.
A similar approach, with additional details, is applicable to computing a $0.5$-approximate DDS as well. 

The described process might take too many peeling steps, e.g., in each peeling step only the smallest-degree vertices are removed.
This might make the algorithm inefficient for computation in settings like semi-streaming, parallel and distributed.
To alleviate that, Bahmani, Kumar, and Vassilvitskii~\cite{bahmani2012densest} proposed an elegant modification in which the peeling is performed by removing \emph{all vertices} with degree below the average degree times $(1+\eps)$.
That enables the authors to guarantee that there are only $O(\log n)$ peeling iterations, and it can also be shown that the algorithm outputs a $(0.5-\eps)$-approximate DS, both undirected and directed.
Their peeling algorithm was also adapted to the first semi-streaming algorithm for DS, taking $O(\log n)$ passes to compute a $(0.5-\eps)$-approximate DS.
More recent peeling-based approaches leading to state-of-the-art results are discussed in \cref{section:related-work}.

To the best of our knowledge, there are no constant-approximate DS peeling algorithms using $o(\log n)$ iterations.
In particular, \cref{fig:peeling-construction} demonstrates a construction where the peeling-only approach reduces the number of vertices in a graph by a constant factor each iteration until it arrives at a constant-factor approximate DS; such an approach inherently takes $\Omega(\log n)$ iterations.

\paragraph{Sampling-based approaches.}
To further the goal of approximating the DDS using sublinear working memory, \cite{esfandiari2016brief} showed that a uniform sample of $\tilde{O}(n\sqrt{n}/\eps^2)$ edges of an input graph $G$ retains a $(1-\eps)$-approximate directed DS in $G$.
It is not hard to provide a construction, e.g., see \cite{pmlr-v235-mitrovic24a} or \cref{fig:construction}, that shows that a uniform sample of $\Omega(n^{1.5 - \beta})$ edges is \textit{needed} to maintain a constant approximation for any $\beta\in (0, 0.5)$.
Therefore, the size of a uniform sample is tight to within an $n^{o(1)}$ factor.

This presents a substantial gap with computing undirected DS, for which it is known that a uniform sample of $\tO(n / \eps^2)$ edges, with high probability, retains a $(1-\eps)$-approximate undirected DS~\cite{esfandiari2016brief}.
Therefore, it remains an open problem whether a directed graph can be sparsified to $\tO(n)$ size, while preserving a $(1-\eps)$-approximate directed DS.

\begin{figure}[H]
    \centering
    \includegraphics[width=0.8\linewidth]{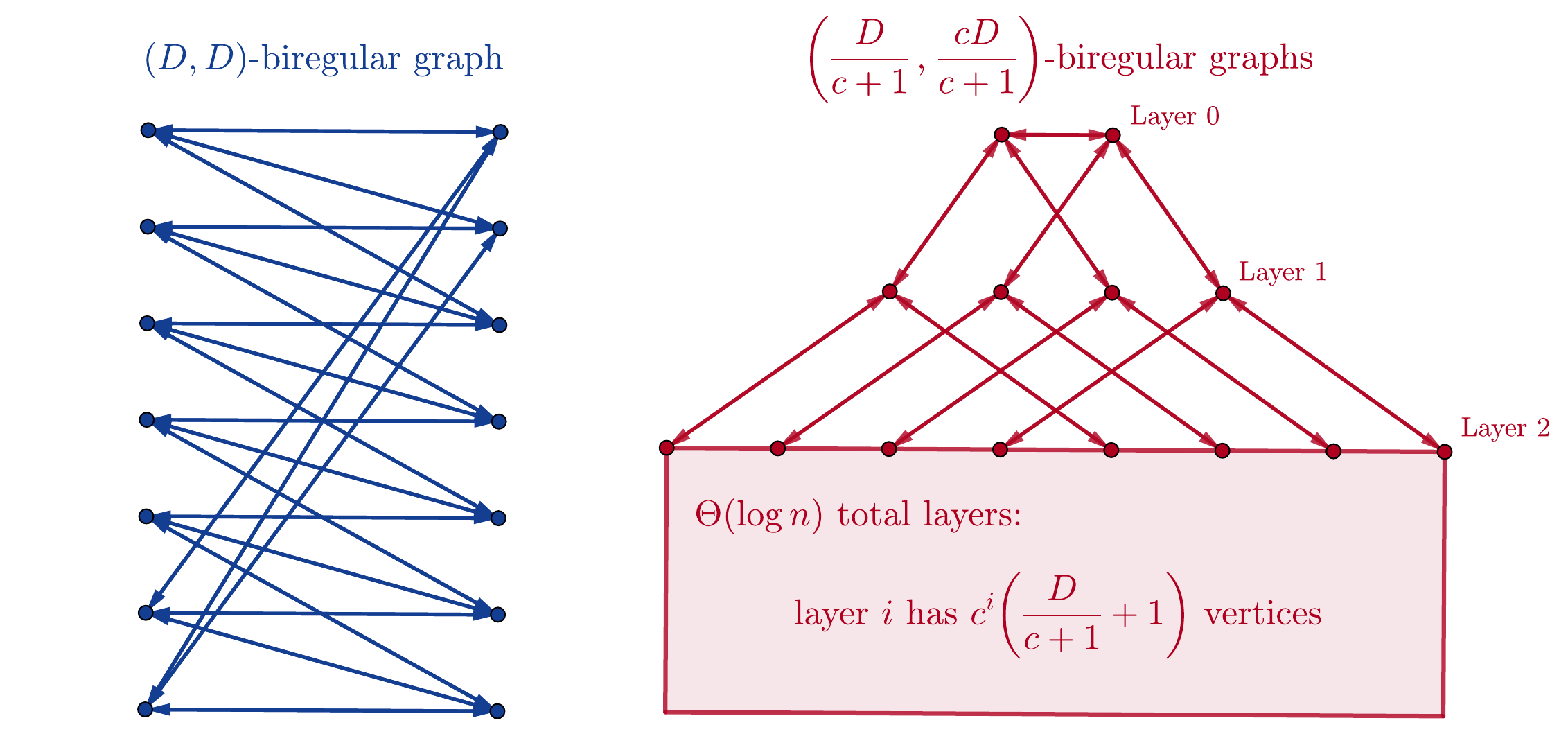}
    \caption{\small
		An example that takes $\Theta(\log n)$ iterations of peeling. Let $c$ be a constant. 
		The subgraph on the left is the DS with density $D$. 
		The subgraph on the right consists of biregular graphs between layers. Layer $0$ induces a $D/(c+1)$ regular graph, while other layers induce empty graphs. The subgraph on the right has maximum density $\frac{\sqrt{c}}{c+1}D$, which is a constant factor smaller than $D$. 
		We observe that all vertices have degree exactly $D$ (including the vertices in the DS) except the bottom layer of the subgraph on the right, where a vertex has degree $\frac{D}{c+1}$. 
		Therefore, any degree-based peeling that preserves the DS will only remove the bottom layer iteratively.
		Since there are $\Theta(\log n)$ layers, this will take $\Theta(\log n)$ iterations, each iteration reducing the graph size by approximately a factor of $c$.
    }
    \label{fig:peeling-construction}
\end{figure}

\begin{figure}[H]
    \centering
    \includegraphics[width=0.85\linewidth]{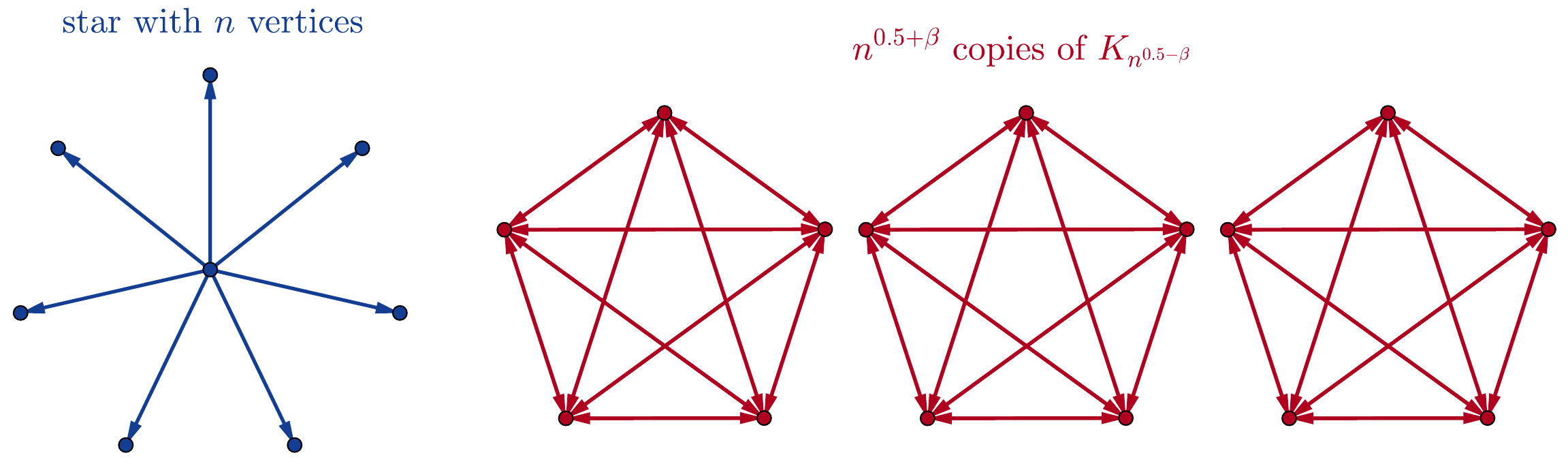}
    \caption{\small An illustration of the construction that requires $\Omega(n\sqrt{n})$ edges of sampling to maintain an approximate directed DS. The graph has $\Theta(n)$ vertices and $\Theta(n^{1.5 -\beta})$ edges for any $\beta \in (0, 0.5)$. The star is the densest subgraph and a uniform sample needs to maintain $\Theta(n)$ edges of the star to still have a constant approximation of the directed DS.
    }
    \label{fig:construction}
\end{figure}

\subsection{Our results}
Our main contribution is showing that a (refined version of) uniform sampling and a \emph{single} step of peeling results in a graph on only $\tO(n / \eps^2)$ edges that preserves a $(1-\eps)$-approximate \emph{directed} DS.
We call this black-box sparsification algorithm \PeelSample, and we attain the following result.
\begin{resultbox}[\PeelSample, Informal version of \cref{lemma:peel-sample-memory}, \cref{theorem:peel-sample-approx}]
There exists a threshold $\tau$ such that after:
\begin{itemize}
    \item Removing some vertices with degree less than $\tau$, and
    \item Sparsifying the remaining graph so that it retains $\tO(n/\eps^2)$ edges,
\end{itemize}
the resulting graph preserves a $(1-\eps)$-approximate DDS.

\end{resultbox}

\subsubsection*{Applications}
We show how to instantiate \PeelSample in semi-streaming, the massively parallel computation model, and the sublinear-time setting, yielding the following results.
\begin{resultbox}[Semi-streaming, two passes, \cref{theorem:twopass}]
There exists a two-pass semi-streaming algorithm using $\tO(n/\eps^2)$ memory that computes a $(1-\eps)$-approximate DDS with probability at least $1 - 1/\poly(n)$.
\end{resultbox}
With additional ideas and allowing for an additional $\eps$ factor in memory, our main approach can also be implemented in a single pass of semi-streaming.
\begin{resultbox}[Semi-streaming, single pass, \cref{theorem:singlepass}]
There exists a single-pass semi-streaming algorithm using $\tO(n/\eps^3)$ memory that computes a $(1-\eps)$-approximate DDS with probability at least $1 - 1/\poly(n)$.
\end{resultbox}
Previously, the only semi-streaming algorithm that computed a constant approximation of the DDS required $\Theta(\log n)$ passes and had a $(0.5-\eps)$-approximation \cite{bahmani2012densest}.
Additionally, the state-of-the-art for single-pass semi-streaming guarantees only an $O(\log n)$-approximation \cite{mitrovic2025new}.
With our result, we settle the approximation problem of computing an approximate directed DS in the single-pass semi-streaming setting, matching the state-of-the-art for undirected DS \cite{esfandiari2016brief}.

Moreover, we improve upon \cite{pmlr-v235-mitrovic24a}, which computed a $(0.5-\eps)$-approximate DDS in $O(\sqrt{\log n})$ MPC rounds in the near-linear memory regime.
Our result reduces the round complexity to a constant.
\begin{resultbox}[Massively Parallel Computation, \cref{theorem:near-linear-MPC}]
There exists a MPC algorithm in the near-linear memory regime using $O(1)$ rounds that computes a $(1-\eps)$-approximate DDS with probability at least $1 - 1/\poly(n)$. The algorithm uses $\tilde{O}(n)$ memory per machine and $\tilde{O}(m)$ total memory.
\end{resultbox}

Finally, the only previous sublinear-time algorithm for DDS comes from generalizing the sampling algorithm in \cite{esfandiari2016brief}, requiring $\tO(n\sqrt{n})$ time, $\tO(n\sqrt{n})$ memory, and $\tO(n\sqrt{n})$ queries.
Our result shaves off a factor of $\sqrt{n}$ in all of the complexities.
\begin{resultbox}[Sublinear-time, \cref{theorem:sublineartime}]
There exists a sublinear-time algorithm that computes a $(1-\eps)$-approximate DDS with probability at least $1 - 1/\poly(n)$.
The algorithm uses $\tilde{O}(n)$ time, $\tilde{O}(n)$ space, and uses $\tilde{O}(n)$ queries to an incident-list oracle model.
\end{resultbox}

\begin{table}[t]
\centering
\renewcommand{\arraystretch}{1.35} 
\setlength{\tabcolsep}{8pt}        

\begin{tabularx}{0.95\textwidth}{|l|c|c|X|}
\hline
Reference & Approximation & Memory & \makecell[l]{Complexity\\\small (e.g., rounds, passes)} \\
\hline

\multicolumn{4}{|c|}{\textbf{Semi-streaming}} \\
\hline
\cite{bahmani2012densest} & $(0.5-\eps)$ & $O(n)$ & $O(\log n)$ passes \\
\hline
\cite{esfandiari2016brief} & $(1-\eps)$ & $\tO(n\sqrt{n}/\eps^2)$ & 1 pass \\
\hline
\cite{pmlr-v235-mitrovic24a} & $(0.5-\eps)$ & $\tO(n/\eps^2)$ & $1$ pass (randomized stream) \\
\hline
\cite{mitrovic2025new} & $O(\log n)$ & $\tO(n)$ & $1$ pass \\
\hline
\textbf{Our work (\cref{theorem:twopass})} & $(1-\eps)$ & $\tO(n/\eps^2)$ & $2$ passes \\
\hline
\textbf{Our work (\cref{theorem:singlepass})} & $(1-\eps)$ & $\tO(n/\eps^3)$ & $1$ pass \\
\hline

\multicolumn{4}{|c|}{\textbf{Massively Parallel Computation}} \\
\hline
\cite{pmlr-v235-mitrovic24a} & $(0.5-\eps)$ & $\tO(n)$ & $O(\sqrt{\log n})$ rounds \\
\hline
\textbf{Our work (\cref{theorem:near-linear-MPC})} & $(1-\eps)$ & $\tO(n)$ & $O(1)$ rounds \\
\hline

\multicolumn{4}{|c|}{\textbf{Sublinear-time}} \\
\hline
\cite{esfandiari2016brief} & $(1-\eps)$ & $\tO(n\sqrt{n})$ & $\tO(n\sqrt{n})$ oracle queries and $\tO(n\sqrt{n})$ time \\
\hline
\textbf{Our work (\cref{theorem:sublineartime})} & $(1-\eps)$ & $\tO(n)$ & $\tO(n)$ oracle queries and $\tO(n)$ time \\
\hline
\end{tabularx}

\caption{Comparison of prior and our work on \textbf{directed} densest subgraph.}
\label{tab:mpc-vs-streaming}
\end{table}

\subsection{Related work}
\label{section:related-work}
The approaches of LPs and peeling, with sampling, are used in many related works and settings such as sublinear MPC, dynamic, \LOCAL, \CONGEST, and differential privacy.
These approaches are not just applied to the DS problem but also for other problems including edge orientation and $k$-core decomposition.

\paragraph{LP-based works}
The LP formulations introduced in 2000~\cite{charikar2000greedy} continue to influence state-of-the-art results.
\cite{bahmani2014efficient} used the MWU framework to approximate the dual LP of the DS problem, and then round that solution to a $(1-\eps)$-approximate DS.
Adapting the way MWU is used in this approach, \cite{ghaffari2019improved} showed how to obtain a $(1-\eps)$-approximation for undirected DS with high probability in $\tO(\sqrt{\log n})$ sublinear MPC rounds.
A key ingredient in the algorithm of \cite{ghaffari2019improved} is to sparsify the graph uniformly so as to reduce the density of the undirected DS to $O(\log n)$; as illustrated above, this is difficult to achieve for directed DS, see, for example, \cref{fig:construction}.

\cite{sawlani2020near} presented a deterministic fully-dynamic algorithm that maintains a $(1-\eps)$-approximate undirected or directed DS with $O(\poly \log n)$ worst-case update and query time.
They developed a novel edge-orientation algorithm that is used to approximately solve the dual LP for DS.

\cite{dinitz2025almost} presented a near-optimal $(\eps,\delta)$-DP algorithm that computes an approximate \textit{undirected} DS with additive error $O(\sqrt{\log(n)\log(n/\delta)}/\eps)$ whp.
They also provided a $(\eps,\delta)$-DP algorithm that computes an approximate \textit{directed} DS with additive error $O(\sqrt{\max(z, 1/z)}\cdot\log(n)\log(n/\delta)/\eps)$ whp, where $z = \sqrt{|S^*|/|T^*|}$.
Their work adapted the MWU method into the differential privacy setting.

\cite{su2020distributed} showed how to attain a $(1-\eps)$-approximate undirected or directed DS given density parameter $D$ using $O(\log n/\eps)$ rounds in the \LOCAL model.
Their algorithm is deterministic and tight in the number of rounds.
They also presented an algorithm that finds a $(1-\eps)$-approximate undirected DS whp given density parameter $D$ using $O(\log^3 n/\eps^3)$ rounds in the \CONGEST model.
Their algorithm requires adapting the MWU method to solve the DS dual LP in the \CONGEST model.

\paragraph{Peeling-based works.}
\cite{charikar2000greedy} introduced the peeling approach for DDS.
Compared to peeling for undirected DS, the algorithm for directed DS is also given ``a guess'' $z$ of $|S^*|/|T^*|$, where $(S^*, T^*)$ is a DDS.\footnote{Observe that $|S^*|, |T^*| \in \{1, \ldots, n\}$, and hence there are at most $n^2$ different guesses $z$. 
Moreover, ``guessing'' $|S^*|/|T^*|$ up to a factor of $(1-\eps)$ incurs only $1-\eps$ factor in the approximation, while reducing the number of guess to $O(\log(n) / \eps)$. It is worth noting that this parameter $z$ has been inspired by the LP formulation from \cite{charikar2000greedy}.}
The peeling-based approach for DDS initializes $S = T = V$, and then in each peeling step it removes vertices either from $S$ or from $T$.
If the current $|S| / |T|$ ratio is at least $z$, the peeling is performed on $S$ and, otherwise, it is performed on $T$.
Since the peeling approach for DS introduced by \cite{charikar2000greedy}, there have been additional works with peeling algorithms that were not mentioned before.
\cite{khuller2009finding} improved on \cite{charikar2000greedy} by removing the need for the parameter $z$ in their peeling algorithm, resulting in the first $2$-approximate DDS algorithm that ran in linear-time.
\cite{pmlr-v235-mitrovic24a} developed an algorithm that finds a $(0.5-\eps)$-approximate DDS in a single pass in a \textit{randomized} stream.

\cite{mitrovic2025new} presented a new way of peeling by leveraging parameters $D = \rho(S^*, T^*)$ and $z = \sqrt{|S^*|/|T^*|}$.
They showed that by repeatedly removing vertices in $S$ with degree less than $D/(2z)$ and vertices in $T$ with degree less than $Dz/2$, they end up with a subgraph where all the vertices have degrees above these thresholds.
As a result, this subgraph is a $2$-approximation of the DDS.
This algorithm, however, can take many iterations, so they introduced conditions that force the vertex sets to decrease by at least a factor of $(1+\eps)$ each iteration, or else they find a $(0.5-\eps)$-approximate DDS.
This results in an algorithm that takes $O(\log n)$ iterations.
Compared to \cite{bahmani2012densest}, the peeling done by \cite{mitrovic2025new} uses fixed thresholds, parametrized by $D$ and $z$, while the threshold for peeling in \cite{bahmani2012densest} is constantly changing, requiring recomputing the average degree over and over.
This property of having fixed thresholds for peeling allows the algorithm in \cite{mitrovic2025new} to be more amenable to adaptation to distributed settings and semi-streaming.
They developed an algorithm that computes a $(0.5-\eps)$-approximate DDS whp in $\tO(\sqrt{\log n})$ sublinear MPC rounds and a deterministic algorithm that computes a $O(\log n)$-approximate DDS in a single pass of semi-streaming.

\cite{bhattacharya2015space} developed a fully-dynamic algorithm that maintains a $(1/4-\eps)$-approximate \textit{undirected} DS whp while using $\tO(n)$ memory and $O(\poly \log n)$ amortized update and query time.
They used l0-samplers to maintain multiple uniform edge samples and approximately compute which iterations of peeling a vertex would be removed in.
They were unable to maintain the sublinear memory and $O(\poly \log n)$ amortized update and query times for the directed DS due to issues with sampling.

\paragraph{Related problems.}
Finally, these approaches for DS have also inspired works in other problems including edge orientation and $k$-core decomposition.
For edge orientation, the goal is to orient all the edges of an undirected graph so that the maximum out-degree of any vertex is as close to $\lceil D^* \rceil$ as possible. \cite{picard1982network} proved that the maximum out-degree of any vertex is always lower bounded by $\lceil D^* \rceil$.
\cite{barenboim2008sublogarithmic} showed that one can compute a $(2+\eps)$-orientation by iteratively peeling all vertices that have degree less than or equal to $(2+\eps)D^*$ and orienting the edges of all the vertices that are removed as out-edges.
They proved that this peeling algorithm always takes $O(\log n)$ iterations.
As a result, this peeling algorithm can be adapted to both \LOCAL and \CONGEST models using $O(\log n)$ rounds and a semi-streaming algorithm that uses $O(\log n)$ passes.
\cite{su2020distributed} improved the approximation in the \CONGEST model, presenting a deterministic algorithm that computes a $(1+\eps)$-orientation using $\tO(\log^2 n/\eps^2)$ rounds.
They accomplished this by observing that the dual LP for DS is the same as edge orientation.
Then, they adapted the MWU method for solving the LP and developed a new rounding procedure that runs efficiently and deterministically in the \CONGEST model.
\cite{ghaffari2019improved} showed how to find a $(2+\eps)$-orientation using $\tO(\sqrt{\log n})$ sublinear MPC rounds by adapting the peeling algorithm above.
Recently, \cite{ghaffari2025density} presented a new sublinear MPC algorithm that computes a $O(\log \log n)$-orientation using $O(\poly \log \log n)$ rounds.
They also adapted the same peeling algorithm but introduced new pruning ideas that keep the neighborhoods of vertices small.

For $k$-core decomposition, there is a natural peeling algorithm to find the $k$-core of a graph: iteratively remove all vertices with degree less than $k$ until the remaining vertices in the graph all have degree at least $k$.
Iteratively applying this peeling algorithm for $k = 1, 2, \ldots, 2D^*$ on the remaining graph from the previous iteration, we attain a full core-labeling of the graph \cite{matula1983smallest}.
\cite{esfandiari2018parallel} adapted this algorithm to a single-pass semi-streaming algorithm that computes a $(1-\eps)$-approximate core-labeling.
They maintained $\Theta(\log n)$ different uniform samples of the graph with different sampling probabilities and used the samples to approximate degrees for different core-labels.
\cite{ghaffari2019improved} adapted the peeling algorithm to the MPC setting, presenting a near-linear MPC algorithm that computes a $(1-\eps)$-approximation whp using $O(\log \log n)$ rounds and a sublinear MPC algorithm that computes a $(0.5-\eps)$-approximation using $\tO(\sqrt{\log n})$ rounds.

Overall, these works show that the LP, peeling, and sampling approaches for DS are applicable to many other settings and problems.

\subsection{An outline of our approach}
\cref{section:sparsification} presents our sparsification algorithm \PeelSample. \cref{section:streaming} shows how to adapt \PeelSample to the semi-streaming setting, introducing further ideas to attain a single-pass algorithm.
We present adaptations of \PeelSample to the near-linear MPC setting in \cref{section:mpc} and to the sublinear-time setting in \cref{section:sublineartime}.

\subsubsection{Refined sampling}
Our main goal is to improve the uniform sampling presented in \cite{esfandiari2016brief} that uses $\tO(n^{1.5}/\eps^2)$ memory to maintain a $(1-\eps)$-approximate DDS.
As noted in the construction above, a uniform sample of $\Omega(n^{1.5 - \beta})$ edges is needed to preserve a constant approximation for any $\beta\in (0, 0.5)$.
However, the construction is a worst-case input graph.
\emph{Is there a way to design an adaptive sampling scheme?}
We observe that the ratio of the size of the vertex sets of the DDS, $|S^*|/|T^*|$, is a useful, common parameter for DS algorithms.
This parameter is inspired by many previous works (\cite{charikar2000greedy}, \cite{bahmani2012densest}, \cite{bahmani2014efficient}, \cite{mitrovic2025new}) which all use the knowledge of this ratio for their approximation algorithms and LPs for DDS.
Letting $z = \sqrt{|S^*|/|T^*|}$, we are able to show in \cref{theorem:sample-approx} that a sample of $O(n\log n\cdot \max(z,1/z)/\eps^2)$ edges is enough to preserve a $(1-\eps)$-approximate DDS.
In doing so, our sampling improves upon the work of \cite{esfandiari2016brief} in the cases where $\max(z, 1/z) = o(\sqrt{n})$.

\emph{Remark:} An algorithm computing a densest subgraph from such a sample should, in a natural way, take into account that the sample was obtained with respect to a given $z$.
Many known algorithms, e.g., \cite{charikar2000greedy, bahmani2012densest,bahmani2014efficient,mitrovic2025new}, are already structured that way, and can be readily used for recovering an approximate DDS from our sample.
More details are discussed in \cref{sec:refined-uniform-sampling}, \cref{section:other-condition}, and \cref{section:algos-condition}.

\subsubsection{Adding low-degree peeling}
Our new refined sampling technique is still not enough to reach our goal of attaining a sparsified graph which has only $\tO(n)$ edges and preserves a constant approximation of a DDS.
We still need to remove the factor of $\max(z, 1/z)$ from the memory usage.
To achieve that, we apply a single iteration of peeling to low-degree vertices in the graph as preprocessing.
We note that peeling algorithms from previous works all use $O(\log n)$ iterations, so it is perhaps surprising that a single iteration of peeling is all we need.
The novel observation we make is that iterations of peeling are not necessarily impactful in decreasing the graph size, but \textit{the first iteration of peeling is special}.
To our knowledge, no previous work with peeling algorithms has analyzed the first iteration of peeling in the following way.

Initially, by representing the directed graph as an undirected bipartite graph (refer to \cref{sec:definitions}), the graph consists of two equal sized vertex sets, having a vertex set ratio of $1$.
The idea behind peeling algorithms for finding the DDS is to leverage the parameter $z^2 = |S^*|/|T^*|$ to keep the vertex set ratio around $z^2$.
If the current ratio is above $z^2$, more vertices in set $S$ should be removed to decrease the ratio, and if the current ratio is below $z^2$, more vertices in set $T$ should be removed to increase the ratio.
Since the initial starting ratio of the vertex sets is $1$, we observe that the first iteration of peeling results in one of the two outcomes:
\begin{itemize}
	\item Reducing the $S$-to-$T$ size ratio to close to $z^2$, forcing one of the vertex sets to drop by a large factor proportional to $\max(z^2, 1/z^2)$.
	\item The resulting subgraph after peeling is very dense and is already a good approximation of the DDS.
\end{itemize}
This is depicted in part (a) of \cref{fig:peeling}.
Therefore, if the resulting subgraph is not dense and the size of one of the vertex sets decreases by a significant amount, we show that the number of edges in the graph must also decrease by a large factor proportional to $\max(z, 1/z)$.
This is formalized in the statement of \cref{lemma:peel-memory}.
Combining this with our previous sampling scheme, the factor of $\max(z, 1/z)$ cancels out in the memory usage.
We also show that this first iteration of peeling does not affect the densest subgraph.
So, after low-degree peeling and sampling, we are left with a sparsified graph that has $\tO(n)$ edges and preserves a $(1-\eps)$-approximate directed DS.
We call this sparsification algorithm \PeelSample and its approximation guarantees are presented in \cref{theorem:peel-sample-approx}.

\begin{figure}[h]
    \centering
    \includegraphics[width=0.9\linewidth]{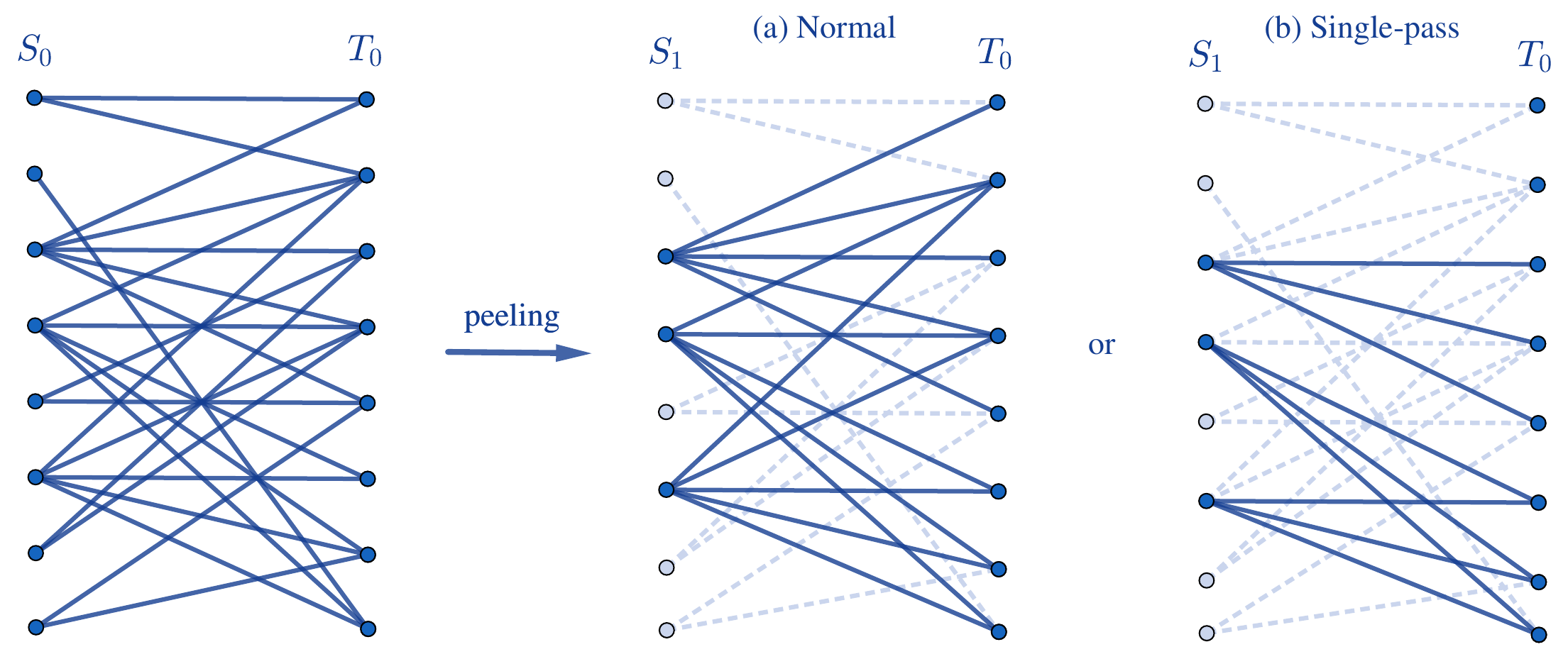}
    \caption{\small An illustration of the first iteration of peeling. Initially, $|S_0|/|T_0| = 1$. 
    (a) Assuming $z^2 < 1$, we would want to decrease the size of set $S_0$. After the first iteration of peeling, all the vertices below the threshold are removed from $S_0$ and the vertices remaining are $S_1$. Then, our claim is that either $|S_1| = O(z^2|S_0|)$ or that $(S_1, T_0)$ is a dense subgraph. (b) During single-pass semi-streaming, vertices remove $\tau - 1$ edges incident to them, where $\tau$ is the threshold. Therefore, the graph is missing a few edges for each vertex that remains in set $S_1$.
    }
    \label{fig:peeling}
\end{figure}

\subsubsection{Single-pass semi-streaming}
Though we can easily adapt our \PeelSample framework to many different settings, it is difficult to adapt it to the single-pass semi-streaming setting.
Note that it is straightforward to simulate \PeelSample in two passes though: the first pass simulates peeling and the second pass simulates sampling.
Then, the goal of a single-pass algorithm would be to combine these two passes, trying to peel vertices at the same time as it samples edges.
However, we observe that this is difficult to accomplish since we do not know which vertices should be peeled while in the middle of a stream.
There is always the possibility that the degree of a vertex that was originally small could suddenly increase past the peeling threshold near the end of the stream.
For example, consider the construction in \cref{fig:construction}.
If all the edges of the complete graphs come before the edges of the star, we will not know which vertices to peel (e.g. all the vertices of the complete graphs) till going through a majority of the stream because their degree could surpass the threshold in the future.
By that time, we would sample $\Omega(n^{1.5 - \beta})$ edges for $\beta\in (0, 0.5)$ since our sampling probability would be a constant to preserve the star graph.
This exceeds the memory bound for the semi-streaming setting.

To not oversample, we need a way to \emph{remove edges without knowing} if a vertex should be peeled during the single iteration of peeling.
To illustrate our idea, assume we are peeling vertices from vertex set $S$ below a threshold $\tau$.
Then, instead of removing/peeling vertices, for each vertex $v \in S$, we remove the first $\tau - 1$ streamed edges incident to $v$.
We observe that all edges which should be removed during the first iteration of peeling (e.g. edges incident to vertices in $S$ with degree less than $\tau$) are immediately removed when seen in the stream and no longer cause us to oversample edges.
However, we also remove $\tau - 1$ edges from each vertex in $S$ that is not removed during peeling.
This is shown in part (b) of \cref{fig:peeling} in comparison to how the first iteration of peeling would normally look in part (a).
The removal of these additional edges affects the edges in the densest subgraph, though we show that with the right parameters of $D$, $z$, and the peeling threshold $\tau$, the density of the densest subgraph does not change too much.
This is formalized in \cref{lemma:singlepass:approx}.
As a result, we compute a $(1 - \eps)$-approximate DDS in a single pass.

\subsection{Open problems}
For our single-pass semi-streaming result, our memory usage is $\tO(n/\eps^3)$ when the memory usage for undirected DS is $\tO(n/\eps^2)$.
So, there still remains a slight gap in the memory usage by a factor of $\eps$.
Generally, we also hope that our \PeelSample framework will find applications to other settings or inspire ideas to other problems for which sampling-only techniques are not sufficiently effective.

\section{Preliminaries}
\subsection{Definitions}
\label{sec:definitions}
\paragraph{Transformation: General directed to bipartite undirected graph.}
Given a directed graph $G_{\text{dir}} = (V_{\text{dir}}, E_{\text{dir}})$, we represent it as a bipartite graph $G = (V_1, V_2, E)$ where: (i) $V_1$ and $V_2$ are two copies of $V_{\text{dir}}$, and (ii) there is an edge in $E$ between $u\in V_1$ and $v\in V_2$ iff there is a directed edge from $u$ to $v$ in $E_{\text{dir}}$.
All directed graphs will be treated as bipartite with this representation. 

\paragraph{Notation.} For a bipartite graph $G = (S, T, E)$, we use $n$ to refer to $|S| + |T|$, the total number of vertices.
Given two vertex sets $A\subseteq S$ and $B\subseteq T$, we refer to the edges between them by $E_G(A, B) \eqd \{e = (i, j) \in E : i\in A, j\in B\}$.
We use $d_G(v)$ to denote the degree of vertex $v$ in $G$.

\paragraph{Directed densest subgraph.} 
Given bipartite graph $G = (V, V, E)$ and vertex sets $S,T\subseteq V$, the \textit{density} $\rho_G(S, T)$ is defined as $\rho_G(S, T) \eqd |E_G(S,T)|/\sqrt{|S| \cdot |T|}$. A \textit{directed densest subgraph} is a pair of sets $S^*$, $T^*$ such that $(S^*, T^*) \in \arg \max_{S,T\subseteq V} \rho_G(S, T)$.

\paragraph{Massively Parallel Computation (MPC).}
In MPC, synchronous rounds of computation are performed across $N$ machines.
Each machine has $S$ words of memory and, initially, the input data is arbitrarily distributed across the machines.
During a round, each machine does computation on its local data.
Then, after the round, machines exchange messages synchronously.
Each machine can send messages to any other machine, but each machine can send and receive at most $S$ words of data.
The primary objective is to perform computation in as few rounds as possible.
With respect to $S$, three regimes are primarily studied: given $\delta \in (0,1)$, \emph{sub-linear} ($S = n^\delta$), \emph{near-linear} ($S = n\poly \log n$), and \emph{super-linear} ($S = n^{1+\delta}$).
Our work focuses on the near-linear memory regime.

\paragraph{Semi-streaming.}
In the semi-streaming setting, an input graph $G = (V, E)$ is given as a stream of edges.
The algorithm receives one edge $e \in E$ at a time and updates its internal memory based on $e$.
This internal memory is constrained to be $O(n \poly \log n)$.
After all edges are presented as a stream, we say the algorithm made a \emph{pass} over the graph. An algorithm can make multiple passes over data.
During a pass, the algorithm can perform arbitrarily large polynomial-time computations.

\paragraph{Sublinear-time.}
In the sublinear-time setting, we are given access to an incident-list oracle model for a directed graph.
In this oracle model, we have two types of queries: (1) asking for the out-degree or in-degree of some vertex $v$, and (2) asking for the $i$th out-neighbor or in-neighbor of some vertex $v$.
We note that degree queries are not necessarily needed as the out-degree or in-degree of a vertex can be calculated using a binary search over the neighbors of a vertex through the neighbor queries.
Therefore, if we ignore a logarithmic factor in the query complexity, we can assume that we only use neighbor queries.
The goal is to produce an algorithm that uses a sublinear number of queries to this oracle and uses sublinear memory and time.

\subsection{Folklore lemmas}

We have the following lower bounds on the degrees of vertices in a DDS.
\begin{lemma}
\label{lemma:existence}
Let $G^* = (S^*, T^*, E_G(S^*, T^*))$ be a directed densest subgraph. Let $D = \rho_G(S^*, T^*)$ be its density and $z = \sqrt{|S^*|/|T^*|}$.
Then it holds that $d_{G^*}(v) \geq D/(2z)$ for all $v\in S^*$ and $d_{G^*}(v) \geq Dz/2$ for all $v\in T^*$. 
\end{lemma}
\begin{proof}
Consider any vertex $v\in S^*$. 
Since $G^*$ is the directed densest subgraph, removing $v$ does not increase the density. Therefore, we have
\begin{eqnarray*}
\frac{|E_G(S^*, T^*)|}{\sqrt{|S^*||T^*|}} &\geq& \frac{|E_G(S^*, T^*)| - d_{G^*}(v)}{\sqrt{(|S^*| - 1)|T^*|}}\\ \implies d_{G^*}(v) &\geq& \rb{1 - \sqrt{\frac{|S^*|-1}{|S^*|}}}|E_G(S^*,T^*)|\\
\implies d_{G^*}(v) &\geq& \frac{|E_G(S^*,T^*)|}{2|S^*|} = \frac{D}{2z}
\end{eqnarray*}
where for the last inequality, we used $\rb{1 - 1/|S^*|}^{1/2} \leq 1 - 1/(2|S^*|)$. Similarly, we can show that for any vertex $v\in T^*$, it must satisfy $d_{G^*}(v) \geq Dz/2$.
Hence, the densest subgraph $(S^*, T^*)$ satisfies the constraints of the claim.
\end{proof}

To argue about sampling, we will use the following Chernoff bound often.
\begin{lemma}[Chernoff bound]\label{lemma:chernoff}
	Let $X_1, \ldots, X_k$ be independent random variables taking values in $[0, 1]$. Let $X \eqdef \sum_{i = 1}^k X_i$ and $\mu \eqdef \E{X}$. Then,
	\begin{enumerate}[(A)]
		\item\label{item:delta-at-most-1} For any $\delta \in [0, 1]$ it holds $\prob{|X - \mu| \ge \delta \mu} \le 2 \exp\rb{- \delta^2 \mu / 3}$.
		\item\label{item:delta-at-least-1} For any $\delta \ge 1$ it holds $\prob{X \ge (1 + \delta) \mu} \le \exp\rb{- \delta \mu / 3}$.
	\end{enumerate}
\end{lemma}

\section{Sparsification framework for directed DS}
\label{section:sparsification}
In this section, we describe our main contribution of a sparsification algorithm, \PeelSample, which allows one to extract a $(1-\eps)$-approximate DDS from the sparsified graph.
The algorithm consists of two steps: (1) peeling low-degree vertices, and then (2) sampling remaining edges.
We first present our uniform sampling scheme which uses $\tO(n\sqrt{n})$ memory in the worst-case.
Then, we demonstrate how introducing low-degree peeling reduces the memory usage to $\tO(n)$ for any directed graph.

\subsection{Refined uniform sampling}
\label{sec:refined-uniform-sampling}

We present a natural, refined, uniform sampling scheme that allows us to preserve an approximate DDS.
Let $D$ be our desired density and $z$ be the square root of our desired vertex sets ratio.
Ideally, we want $D = \rho_G(S^*, T^*)$ and $z = \sqrt{|S^*|/|T^*|}$, but approximations of these values suffice too.
Our sampling refinement uses $D$ and $z$ to (potentially) decrease the number of sampled edges when the graph does not have lopsided dense subgraphs, e.g. a star.
Specifically, it matches the upper bound of $\tilde{O}(n\sqrt{n})$ edges to a $\poly \log n$ factor only when $\max(z, 1/z) = \Omega(\sqrt{n})$, and samples fewer edges otherwise.
Our \Sample procedure is presented as \cref{alg:sampling}.

\begin{algorithm}
\caption{(\Sample) Computes a uniform sample of edges in graph $G$ using parameters $D, z$}
\label{alg:sampling}
\textbf{Input:} bipartite graph $G = (S, T, E)$, parameters $\eps$, $\alpha$, $D$ and $z$
\begin{algorithmic}[1]
\State $p \gets \min\rb{1, \frac{60\alpha\log n}{\eps^2 D}\max\rb{z, 1/z}}$
\State $E'\gets$ uniform sample of edges in $E$ each with prob. $p$
\State \Return $H = (S, T, E')$
\end{algorithmic}
\end{algorithm}

Let $H$ be the graph output by \Sample. 
The parameter $\alpha$ specifies the approximation of the DDS of $G$ we want to extract from $H$.
Note that $H$ is obtained from $G$ as a function of -- among other -- parameters $D$ and $z$.
Naturally, an algorithm $\cA$ computing an approximate DDS of $G$ while given $H$, might also need to take into account $D$ and $z$.
So, we define the following condition that the subgraph returned by $\cA$ must satisfy to successfully compute an approximate DDS of $G$ from $H$:
\begin{align}
   \condition:\quad  & \text{subgraph }(S', T')\text{ satisfies } \nonumber \\
    & |E_H(S', T')| \geq \frac{pD}{2\alpha}\rb{|S'|/z + |T'|z}.\nonumber
\end{align}
In \cref{section:other-condition} and \cref{section:algos-condition}, we show that many approximate-DDS algorithms from prior work satisfy this condition.
\emph{Remark:} For the sake of brevity, we also say that an algorithm satisfies \condition to imply that the output of the algorithm satisfies \condition.

What does satisfying \condition imply?
One can observe that if $\sqrt{|S'||T'|}$ is fixed, e.g., the denominator of the density function, then the term $|S'|/z + |T'|z$ is minimized when $|S'|/|T'| = z^2$.
Therefore, \condition gets harder to satisfy the further the ratio $|S'|/|T'|$ is from $z^2$.
In this way, \condition can be seen as a gradual condition that at the core wants $|S'|/|T'| \approx z^2$, which fits our intuition, but allows $|S'|/|T'|$ to be further from $z^2$ as long as the subgraph makes up for it by having more edges.
Finally, we have that
\[|E_H(S', T')| \geq \frac{pD}{2\alpha}\rb{|S'|/z + |T'|z} \geq \frac{pD}{\alpha}\sqrt{|S'||T'|}\]
using the AM-GM inequality.
Therefore, $(S', T')$ has density at least $pD/\alpha$ in $H$.
The goal is then to show that $(S', T')$ has density at least $(1-\eps)D/\alpha$ in $G$ with high probability.
As a result, if $\mathcal{A}$ satisfies \condition, then $\mathcal{A}$ can be applied to the sampled graph produced by $\Sample$ to compute an $(1-\eps)/\alpha$-approximate DDS with high probability.
Now, we formally prove the following guarantees about \Sample.

\begin{lemma}[Bound on the number of sampled edges]
\label{lemma:sample-memory}
Consider graph $G = (S, T, E)$, $\eps\in(0,1/2)$, $\alpha \geq 1$, $D \leq m$, and $z\in[1/\sqrt{n}, \sqrt{n}]$. Let $H = \Sample(G, \eps, \alpha, D, z)$. Then, with probability at least $1 - n^{-6}$, the number of edges in $H$ is $O(\alpha m\log n\cdot\max(z, 1/z)/(\eps^2D))$.
\end{lemma}
\begin{proof}
The expected number of edges in $H$ is
\[\E{|E(H)|} = pm \leq \frac{60\alpha m\log n}{\eps^2 D} \max\rb{z, 1/z}.\]
Therefore, using the Chernoff Bound, we see that
\[\prob{|E(H)| - \frac{60\alpha m\log n}{\eps^2 D} \max\rb{z, 1/z} \geq \eps \frac{60\alpha m\log n}{\eps^2 D} \max\rb{z, 1/z}} \leq n^{-6}.\]
The number of edges in $H$ is $O(\alpha m\log n\cdot\max(z, 1/z)/(\eps^2D))$ with probability at least $1 - n^{-6}$.
\end{proof}

Since we know that the DDS has density lower bounded by the average degree of the graph, we can assume that $D \geq m/n$ which gives us $O(\alpha n \log n \cdot \max(z, 1/z)/\eps^2)$ edges in the sampled graph by \cref{lemma:sample-memory}.

\begin{lemma}[Density preservation for the DDS]
\label{lemma:sample-DS}
Consider graph $G = (S, T, E)$, $\eps\in(0,1/2)$, $\alpha \geq 1$, $D \geq m/n$, and $z\in[1/\sqrt{n}, \sqrt{n}]$.
Let $H = \Sample(G, \eps, \alpha, D, z)$ be the sampled graph from $G$ with probability $p = \min\rb{1, \frac{60\alpha\log n}{\eps^2 D}\max\rb{z, 1/z}}$.
Then, if $D \leq \rho_G(S^*, T^*)$, the density of $(S^*, T^*)$ in $H$ is at least $(1-\eps)pD$ with probability at least $1 - n^{-6}$.
\end{lemma}
\begin{proof}
We have that
\[\E{|E_H(S^*, T^*)|} = p|E_G(S^*, T^*)| = pD\sqrt{|S^*||T^*|}.\]
Then, using the Chernoff Bound, we see
\[\prob{||E_H(S^*, T^*)| - pD\sqrt{|S^*||T^*|}| \geq \eps pD\sqrt{|S^*||T^*|}} \leq n^{-6}\]
by plugging in $p$.
Therefore, we prove that $(S^*, T^*)$ in $H$ has density at least $(1-\eps)pD$ with probability at least $1 - n^{-6}$.
\end{proof}

\begin{lemma}[\condition excludes sparse subgraphs]
\label{lemma:sample-condition}
Consider graph $G = (S, T, E)$, $\eps\in(0,1/2)$, $\alpha \geq 1$, $D \geq m/n$, and $z\in[1/\sqrt{n}, \sqrt{n}]$.
Let $H = \Sample(G, \eps, \alpha, D, z)$ be the sampled graph from $G$ with probability $p = \min\rb{1, \frac{60\alpha\log n}{\eps^2 D}\max\rb{z, 1/z}}$.
Then, for all subgraphs $(S', T')$ with $\rho_G(S', T') < (1-\eps)D/\alpha$, $(S', T')$ does not satisfy \condition with probability at least $1 - n^{-6}$.
\end{lemma}
\begin{proof}
Without loss of generality, we assume that $z \geq 1$.
If $p = 1$, it is easy to see that everything holds.
Therefore, we assume that $p = \frac{60\alpha\log n}{\eps^2 D}\max\rb{z, 1/z}$.

Consider fixed $s,t$ with $1 \leq s, t \leq n$ and $|S'| = s$, $|T'| = t$.
Now, if $\rho_G(S', T') < (1-\eps)D/\alpha$, we have that
\[\E{|E_H(S', T')|} = p|E_G(S', T')| \leq \frac{(1-\eps)pD}{\alpha}\sqrt{st} \leq \frac{(1-\eps)pD}{2\alpha}\rb{s/z + tz}.\]
Then,
\begin{eqnarray*}
\prob{|E_H(S', T')| - \frac{(1-\eps)pD}{2\alpha}\rb{s/z + tz} \geq \frac{\eps pD}{2\alpha}\rb{s/z + tz}} &\leq&\exp\rb{-\eps^2 pD\rb{s/z + tz}/(6\alpha)}\\
&\leq& n^{-5(s + tz^2)}
\end{eqnarray*}
using the Chernoff bound and plugging in the value of $p$.
Observe that the LHS contains the condition that $(S', T')$ satisfies \condition.
Finally, using the union bound over all $s,t$ and vertex sets $(S', T')$, we have the upper bound
\begin{eqnarray*}
    \sum_{s=1}^n \sum_{t=1}^n \binom{n}{s}\binom{n}{t} n^{-5(s + tz^2)} \leq \sum_{s=1}^n \sum_{t=1}^n n^{-4(s + t)} \leq n^{-6}.
\end{eqnarray*}
Therefore, all such subgraphs $(S', T')$ do not satisfy \condition with probability at least $1 - n^{-6}$.
\end{proof}

Now, we can prove the following structural property of the sampled graph returned by \Sample.
\begin{theorem}
\label{theorem:sample-approx}
Consider graph $G = (S, T, E)$, $\eps\in(0,1/2)$, $\alpha \geq 1/(1-\eps)$, $D \geq m/n$, $z\in[1/\sqrt{n}, \sqrt{n}]$, and $p = \min\rb{1, \frac{60\alpha\log n}{\eps^2 D}\max\rb{z, 1/z}}$.
Let $H = \Sample(G, \eps, \alpha, D, z)$.
Then, for $D = \rho_G(S^*, T^*)$ and $z = \sqrt{|S^*|/|T^*|}$, there exists a subgraph of $H$ that satisfies \condition and has density at least $(1-\eps)D/\alpha$ in $G$ with probability at least $1 - n^{-5}$.
\end{theorem}
\begin{proof}
If $D = \rho_G(S^*, T^*)$ and $z = \sqrt{|S^*|/|T^*|}$, we show that a subgraph exists that satisfies \condition with high probability.
Considering $E_H(S^*, T^*)$, we have with probability at least $1 - n^{-6}$ that
\[|E_H(S^*, T^*)| \geq (1-\eps)pD\sqrt{|S^*||T^*|} = \frac{(1-\eps)pD}{2}\rb{|S^*|/z + |T^*|z}\]
using \cref{lemma:sample-DS}.
So, $(S^*, T^*)$ satisfies \condition since $\alpha \geq 1/(1-\eps)$.
Then, from \cref{lemma:sample-condition}, the density of a subgraph that satisfies \condition is at least $(1-\eps)D/\alpha$ in $G$ with probability at least $1 - n^{-6}$.
Therefore, the theorem holds with probability at least $ 1 - n^{-5}$.
\end{proof}

In other words, there exist algorithms that satisfy \condition and can extract an approximate DDS from the sampled graph while only incurring a $(1-\eps)$ factor in approximation.
This implies that \Sample preserves a $(1-\eps)$-approximate DDS.
It is also not hard to see that \cref{theorem:sample-approx} can be extended to allowing approximations of $D$ and $z$.
As a result, an algorithm could run polylogarithmic instances of $D$ and $z$ to find the parameters that compute the best approximation (largest $D$).
Later sections demonstrate how to use approximations of $D$ and $z$ with \cref{theorem:sample-approx} to prove the approximation guarantees of their algorithms.
We also note that many algorithms from previous works for DDS \cite{charikar2000greedy, bahmani2012densest, bahmani2014efficient, mitrovic2025new} satisfy \condition.
As examples in \cref{section:other-condition} and \cref{section:algos-condition}, we demonstrate how they satisfy \condition and can be used to compute an approximate DDS from the sparsified graph.

\subsection{Combining low-degree peeling with sampling}
Previously, we showed how to maintain a $(1-\eps)$-approximate DDS with a sample of $O(\alpha n\log n\cdot \max(z, 1/z)/\eps^2)$ edges where $z = \sqrt{|S^*|/|T^*|}$.
However, since the factor of $\max(z, 1/z)$ can be $\Theta(\sqrt{n})$, we match the upper bound of $\tilde{\Theta}(n\sqrt{n})$ edges.
Therefore, to reduce the memory usage further, we preprocess the graph by applying low-degree peeling to it.
Let $c \leq 1/2$ be a parameter.
Then, if $z \geq 1$, we remove all vertices in $T$ with degree less than $cDz$, and if $z < 1$, we remove all vertices in $S$ with degree less than $cD/z$.
Observe that this peeling is done \textit{only on one side}, either on $S$ or on $T$.
We then combine this low-degree peeling with our sampling algorithm to get \PeelSample (\cref{alg:peel-sample}).

\emph{Remark:} We note that when $c = 1/2$, it reflects the peeling performed in \cite{mitrovic2025new} which uses $D/(2z)$ and $Dz/2$ as the peeling thresholds.
This parameter $c$ gives us the flexibility to use smaller peeling thresholds that will help attain a $(1-\eps)$-approximate DDS for our single-pass semi-streaming algorithm in \cref{section:streaming}.

\begin{algorithm}
\caption{(\PeelSample) Sparsifies graph $G$ while maintaining approximate DDS}
\label{alg:peel-sample}
\textbf{Input:} bipartite graph $G = (S, T, E)$, parameters $\eps$, $\alpha$, $c$, $D$ and $z$
\begin{algorithmic}[1]
\State $k_S \gets cD/z, k_T \gets cDz$
\If{$z\geq 1$}
    remove all vertices $v\in T$ with $d_G(v) < k_T$
\Else{}
    remove all vertices $v\in S$ with $d_G(v) < k_S$
\EndIf
\If{$\rho(S,T)\geq D$\label{line:peel-sample:den-check}}
    \Return $(S, T)$
\EndIf
\State \Return $\Sample(G, \eps, \alpha, D, z)$
\end{algorithmic}
\end{algorithm}

How do we guarantee that the number of edges in the graph drops substantially by a single iteration of peeling so that $\Sample(G, \eps, \alpha, D, z)$ only has $\tilde{O}(n/\eps^2)$ edges?
This guarantee is due to line $4$ of \cref{alg:peel-sample}.
We show in \cref{lemma:peel-memory} that if the number of edges removed is not large enough, we can immediately compute a subgraph with density at least $D$.

\begin{lemma}[Edge bound after one peeling step]
\label{lemma:peel-memory}
Consider graph $G = (S, T, E)$, $\eps\in(0,1/2)$, $\alpha \geq 1$, $c\in(0, 1/2]$, $D \geq m/n$, and $z\in[1/\sqrt{n}, \sqrt{n}]$. The number of edges remaining in $G$ after line $4$ of $\PeelSample(G, \eps, \alpha, c, D, z)$ is $O\rb{\frac{nD}{c \cdot \max(z, 1/z)}}$ if line $4$ does not return early.
\end{lemma}
\begin{proof}
Without loss of generality, we assume that $z \geq 1$.
Initially, both $S$ and $T$ are copies of $V$. 
So, $|S|/|T| = 1 \leq z^2$.
Therefore, peeling will remove vertices from $T$.
We observe that if line $4$ does not return, then afterwards we have that $|T| < \frac{|V|}{c^2z^2}$. 
To see that, assume that $|T| \geq \frac{|V|}{c^2z^2}$. In that case, using that $S = V$, we have that
\[
	\rho(S,T) = \frac{|E(S,T)|}{\sqrt{|S||T|}} \geq \frac{k_T|T|}{\sqrt{|S||T|}} \geq \frac{k_T}{cz} = D,
\]
which contradicts line $4$ not returning.
Therefore,
$|E(S,T)| < D\sqrt{|S||T|} < \frac{D|V|}{cz}.$
Similarly, we can show that if $z < 1$, the number of edges is bounded by $Dz|V|/c$.
\end{proof}
Now, we present the memory and approximation guarantees of \PeelSample.
To attain approximation guarantees, we show that the single iteration of peeling does not affect the edges of the DDS when $D = \rho_G(S^*, T^*)$ and $z = \sqrt{|S^*|/|T^*|}$.
\begin{lemma}
\label{lemma:peel-sample-memory}
Consider graph $G = (S, T, E)$, $\eps\in(0,1/2)$, $\alpha \geq 1$, $c\in(0, 1/2]$, $D \geq m/n$, and $z\in[1/\sqrt{n}, \sqrt{n}]$.
Then, $\PeelSample(G, \eps, \alpha, c, D, z)$ either returns a subgraph with density at least $D$ or it returns a sampled graph $H$ of $G$ with $O(\alpha n\log n /(c\eps^2))$ edges with probability at least $1 - n^{-6}$.
\end{lemma}
\begin{proof}
If $\PeelSample(G, \eps, \alpha, c, D, z)$ returns early on line $4$, then we compute a subgraph with density at least $D$ and we are done.
Otherwise, from \cref{lemma:peel-memory}, we have that the number of remaining edges in $G$ is $O\rb{\frac{nD}{c\cdot\max(z, 1/z)}}$.
Therefore, the sampled graph $H$ has $O(\alpha n\log n /(c\eps^2))$ edges with probability at least $1 - n^{-6}$ by using \cref{lemma:sample-memory} and observing that the factor of $D/\max(z, 1/z)$ cancels out.
\end{proof}
\begin{theorem}
\label{theorem:peel-sample-approx}
Consider graph $G = (S, T, E)$, $\eps\in(0,1/2)$, $\alpha \geq 1/(1-\eps)$, $c\in(0, 1/2]$, $D \geq m/n$, and $z\in[1/\sqrt{n}, \sqrt{n}]$.
Let $H$ be the sampled graph returned by $\PeelSample(G, \eps, \alpha, c, D, z)$.
Then, for $D = \rho_G(S^*, T^*)$ and $z = \sqrt{|S^*|/|T^*|}$, there exists a subgraph of $H$ that satisfies \condition and has density at least $(1-\eps)D/\alpha$ in $G$ with probability at least $1 - n^{-5}$.
\end{theorem}
\begin{proof}
We first show that there exists a subgraph that satisfies \condition after peeling when
$D = \rho(S^*, T^*)$ and $z = \sqrt{|S^*|/|T^*|}$.
Let us consider the subgraph $(S^*, T^*)$.
From \cref{lemma:existence}, we know that $G^* = (S^*, T^*, E_G(S^*, T^*))$ satisfies $d_{G^*}(v) \geq D/(2z)$ for all $v\in S^*$ and $d_{G^*}(v) \geq Dz/2$ for all $v\in T^*$.
Therefore, the single iteration of low-degree peeling does not remove any of the vertices in $(S^*, T^*)$ since $c\leq 1/2$ and also does not remove any of the edges in $G^*$.
Then, following the argument in \cref{theorem:sample-approx}, we have that there exists a subgraph of $H$ that satisfies \condition and has density at least $(1-\eps)D/\alpha$ in $G$ with probability at least $1 - n^{-5}$.
\end{proof}
Therefore, from \cref{lemma:peel-sample-memory} and \cref{theorem:peel-sample-approx}, we have shown that \PeelSample is able to sparsify a graph so that it only has $O(n\log n/\eps^2)$ edges and preserves a $(1-\eps)$-approximate DDS.

\subsection{Peeling-based algorithms satisfy \condition}
\label{section:other-condition}
For an algorithm to extract an approximate DDS from a sparsified graph returned by \PeelSample, it needs to satisfy \condition.
In this section, we demonstrate examples of peeling-based algorithms from \cite{charikar2000greedy, bahmani2012densest, mitrovic2025new} that all satisfy \condition and can be used to compute an approximate DDS from our sparsified graph.

\subsubsection{Greedy peeling algorithm}
The algorithm from \cite{charikar2000greedy} iteratively removes the vertex in either $S$ or $T$ with the minimum degree.
If $|S|/|T| \geq z^2$, then the vertex is removed from $S$.
Otherwise, the vertex is removed from $T$.
Then, they select the densest subgraph among all the intermediate subgraphs, and they prove that it is a $0.5$-approximate DDS when $z = \sqrt{|S^*|/|T^*|}$.
However, since we want the algorithm to leverage the parameter $D$, we change how it selects the approximate DDS slightly.
We present this in \cref{alg:greedy}.
\begin{algorithm}
\caption{Greedy peeling algorithm for $0.5$-approximate DDS}
\label{alg:greedy}
\textbf{Input:} bipartite graph $G = (S, T, E)$, parameters $D$ and $z$
\begin{algorithmic}[1]
\While{$S\neq \emptyset$ and $T\neq \emptyset$}
    \If{$|S|/|T| \geq z^2$}
        \State $i_\text{min}\gets$ vertex $i\in S$ that minimizes $|E(\{i\}, T)|$
        \If{$|E(\{i_\text{min}\}, T)| \geq D/(2z)$}
            \Return $(S, T)$
        \EndIf
        \State Remove $i_\text{min}$ from $S$
    \Else
        \State $j_\text{min}\gets$ vertex $j\in T$ that minimizes $|E(S, \{j\})|$
        \If{$|E(S, \{j_\text{min}\})| \geq Dz/2$}
            \Return $(S, T)$
        \EndIf
        \State Remove $j_\text{min}$ from $T$
    \EndIf
\EndWhile
\State \Return $(\emptyset, \emptyset)$
\end{algorithmic}
\end{algorithm}

Instead of choosing the subgraph with the largest density among all the intermediate subgraphs, lines $4$ and $8$ leverage the parameter $D$ to select the intermediate subgraph which will satisfy \condition.
Now, we have the following lemma about \cref{alg:greedy} satisfying \condition.
\begin{lemma}
\label{lemma:greedy}
Consider graph $G = (S, T, E)$, $\eps\in(0,1)$, $D \geq m/n$, and $z\in[1/\sqrt{n}, \sqrt{n}]$.
If \cref{alg:greedy} returns a nonempty subgraph $(S', T')$, it satisfies
\[|E_G(S', T')| \geq \frac{D}{4}\rb{|S'|/z + |T'|z}.\]
When $D = \rho_G(S^*, T^*)$ and $z = \sqrt{|S^*|/|T^*|}$, \cref{alg:greedy} never returns an empty subgraph.
\end{lemma}
\begin{proof}
Without loss of generality, let us assume that \cref{alg:greedy} returns a nonempty subgraph $(S', T')$ on line $4$.
Since $i_\text{min}$ has degree at least $D/(2z)$ in subgraph $(S', T')$ and is also the vertex with minimum degree in $S'$, then we have that
\[|E_G(S', T')| \geq |S'| \cdot \frac{D}{2z} \geq \frac{D}{4}\rb{|S'|/z + |T'|z}\]
where we used $|S'|/|T'| \geq z^2$.
Now, if $D = \rho_G(S^*, T^*)$ and $z = \sqrt{|S^*|/|T^*|}$, we know that $G^* = (S^*, T^*, E_G(S^*, T^*))$ satisfies $d_{G^*}(v) \geq D/(2z)$ for all $v\in S^*$ and $d_{G^*}(v) \geq Dz/2$ for all $v\in T^*$ from \cref{lemma:existence}.
Therefore, when the first vertex in $(S^*, T^*)$ is removed, it will have degree at least $D/(2z)$ if it is in $S^*$ and degree at least $Dz/2$ if it is in $T^*$.
As a result, \cref{alg:greedy} will return a nonempty subgraph from line $4$ or line $8$.
\end{proof}

Therefore, the peeling algorithm from \cite{charikar2000greedy} satisfies \condition with $\alpha = 2$.
Additionally, when $D = \rho_G(S^*, T^*)$ and $z = \sqrt{|S^*|/|T^*|}$, running \cref{alg:greedy} on the sparsified graph returned by $\PeelSample(G, \eps, \alpha, D, z)$ computes a $(0.5 - \eps)$-approximate DDS whp using \cref{lemma:greedy} and \cref{theorem:peel-sample-approx}.
The peeling algorithm from \cite{bahmani2012densest} follows very closely with this greedy algorithm and a similar argument shows that it satisfies \condition with $\alpha = 2/(1-\eps)$.
As a result, it also computes a $(0.5-\eps)$-approximate DDS whp from a sparsified graph returned by \PeelSample.

\subsubsection{Fixed threshold peeling algorithm}
The peeling algorithm from \cite{mitrovic2025new} uses a fixed threshold for the degrees of vertices when peeling.
Specifically, during each iteration, all vertices in $S$ with degree less than $D/(2z)$ and all vertices in $T$ with degree less than $Dz/2$ are removed.
We present this in \cref{alg:fixed-peeling}.
\begin{algorithm}
\caption{Fixed threshold peeling algorithm for $(0.5-\eps)$-approximate DDS}
\label{alg:fixed-peeling}
\textbf{Input:} bipartite graph $G = (S, T, E)$, parameters $\eps$, $D$ and $z$
\begin{algorithmic}[1]

\State $k_S \gets D/(2z)$, $k_T \gets Dz/2$

\While{true}

    \State $A\gets$ all vertices $v\in S$ with $d_{G}(v) < k_S$
    
    \State $B\gets$ all vertices $v\in T$ with $d_{G}(v) < k_T$

    \If{$|S|/|T| \geq z^2$ and $|A| \leq \eps|S|$}
        \Return $(S, T)$
    \EndIf
    \If{$|S|/|T| \leq z^2$ and $|B| \leq \eps|T|$}
        \Return $(S, T)$
    \EndIf

    \State Remove $A$ from $S$ and $B$ from $T$

\EndWhile
\end{algorithmic}
\end{algorithm}

Then, we have the following lemma about \cref{alg:fixed-peeling} satisfying \condition.
\begin{lemma}
Consider graph $G = (S, T, E)$, $\eps\in(0,1)$, $D \geq m/n$, and $z\in[1/\sqrt{n}, \sqrt{n}]$.
If \cref{alg:fixed-peeling} returns a nonempty subgraph $(S', T')$, it satisfies
\[|E_G(S', T')| \geq \frac{(1-\eps)D}{4}\rb{|S'|/z + |T'|z}.\]
When $D = \rho_G(S^*, T^*)$ and $z = \sqrt{|S^*|/|T^*|}$, \cref{alg:fixed-peeling} never returns an empty subgraph.
\end{lemma}
\begin{proof}
Without loss of generality, let us assume that \cref{alg:fixed-peeling} returns a nonempty subgraph $(S', T')$ on line $5$.
Then, we observe that at least $(1-\eps)|S'|$ vertices in $S'$ have degree at least $k_S$ in the subgraph $(S', T')$.
Therefore,
\[|E_G(S', T')| \geq (1-\eps)|S'| \cdot k_S = \frac{(1-\eps)D}{2} \rb{|S'|/z} \geq \frac{(1-\eps)D}{4}\rb{|S'|/z + |T'|z}\]
where we used $|S'|/|T'| \geq z^2$.
Now, if $D = \rho_G(S^*, T^*)$ and $z = \sqrt{|S^*|/|T^*|}$, we know that $G^* = (S^*, T^*, E_G(S^*, T^*))$ satisfies $d_{G^*}(v) \geq D/(2z)$ for all $v\in S^*$ and $d_{G^*}(v) \geq Dz/2$ for all $v\in T^*$ from \cref{lemma:existence}.
So, the vertices in $S^*$ and $T^*$ will never be removed during peeling and \cref{alg:fixed-peeling} must return a nonempty subgraph.
\end{proof}
From this lemma, we see that the peeling algorithm from \cite{mitrovic2025new} satisfies \condition with $\alpha = 2/(1-\eps)$.
Similarly to the previous peeling algorithms, \cref{alg:fixed-peeling} computes a $(0.5-\eps)$-approximate DDS whp from a sparsified graph returned by \PeelSample.

\section{Semi-streaming algorithms}
\label{section:streaming}
In this section, we adapt the \PeelSample algorithm to the semi-streaming setting.
A natural implementation leads to a two-pass algorithm but to attain a single-pass algorithm, we require additional new ideas. Specifically, we modify the single iteration of peeling so that it can be done at the same time as sampling, allowing \PeelSample to be approximately simulated in a single pass.

\subsection{Warm-up: Two-pass algorithm}
Directly simulating \PeelSample results in a two-pass semi-streaming algorithm that computes a $(1-\eps)$-approximate DDS.
The first pass simulates low-degree peeling and the second pass simulates sampling.
We present this in \cref{alg:twopass}.

\begin{algorithm}
\caption{Two-pass semi-streaming algorithm for $(1-\eps)$-approximate DDS}
\label{alg:twopass}
\textbf{Input:} vertex sets $(S, T)$, parameters $\eps$, $D$ and $z$
\begin{algorithmic}[1]
\State $k_S \gets D/(2z), k_T \gets Dz/2$
\Statex
\State \textbf{Pass 1 (Low-Degree Peeling)}
\State $d(v)\gets 0$ for all $v \in V$
\While{stream is not empty}
    \State $(u, v) \gets $ next edge from the stream
    \If{$z \geq 1$}
        $d(v) = d(v) + 1$
    \Else{}
        $d(u) = d(u) + 1$
    \EndIf
\EndWhile
\If{$z\geq 1$}
    remove all vertices $v\in T$ with $d(v) < k_T$
\Else{}
    remove all vertices $v\in S$ with $d(v) < k_S$
\EndIf
\If{$\rho(S,T)\geq D$}\label{line:twopass:den-check}
    \Return $(S, T)$
\EndIf
\Statex
\State \textbf{Pass 2 (Sampling)}
\State $H\gets \emptyset$, $\alpha \gets 1/(1-\eps)$
\State $p \gets \min\rb{1, \frac{60\alpha\log n}{\eps^2 D}\max\rb{z, 1/z}}$
\While{stream is not empty}
    \State $(u, v) \gets $ next edge from the stream
    \If{$u\in S$ and $v\in T$}
        add $(u, v)$ to $H$ with prob. $p$
    \EndIf
\EndWhile
\State $(S', T')\gets$ densest subgraph of $H$ that satisfies \condition, or $(\emptyset, \emptyset)$ if no subgraph exists
\State \Return $(S', T')$
\end{algorithmic}
\end{algorithm}

Note that in \cref{alg:twopass}, calculating $\rho(S,T)$ does not need another pass since it can be directly calculated by subtracting the degree of all vertices that have been removed from the total number of edges.
This is true because vertices are only removed from one side, either from $S$ or from $T$.
Now, we present the full semi-streaming algorithm in the following theorem.

\begin{theorem}
\label{theorem:twopass}
For $\eps\in(0, 1/4)$, there exists a two-pass semi-streaming algorithm using $\tO(n/\eps^2)$ memory that computes a $(1-4\eps)$-approximate DDS with probability at least $1 - n^{-4}$.
\end{theorem}
\begin{proof}
Let $\Astream$ be our two-pass semi-streaming algorithm. We now describe it and provide its analysis.

\paragraph{Algorithm description.}
$\Astream$ will run \cref{alg:twopass} on all $m/n \leq D\leq n$ and $\frac{1}{\sqrt{n}} \leq z \leq \sqrt{n}$ in parallel using powers of $(1+\eps)$.
Then, out of all the outputs of \cref{alg:twopass} where the vertex sets are both non-empty, the one corresponding to the largest $D$ is the final output of $\Astream$.

\paragraph{Algorithm memory.}
During the first pass of \cref{alg:twopass}, we only use $O(n)$ memory.
During the second pass, we have that the sample $H$ uses $O(n\log n/\eps^2)$ memory from \cref{lemma:peel-sample-memory} with probability at least $1 - n^{-6}$.
Therefore, over all $O(\log^2 n)$ parallel instances of \cref{alg:twopass}, $\Astream$ uses $O(n\log^3 n/\eps^2)$ memory in total.

\paragraph{Algorithm approximation.}
For all parallel instances of $D$ and $z$, if \cref{alg:twopass} returns early on line $10$, a subgraph with density at least $D$ in $G$ is computed. Otherwise, we have the guarantee that the subgraph returned, if it is not empty, has density at least $(1-2\eps)D$ in $G$ with probability at least $1 - n^{-5}$ from \cref{lemma:sample-condition}.
Now, we know there exists an instance where
\[(1-2\eps)\rho_G(S^*, T^*) \leq D \leq (1-\eps)\rho_G(S^*, T^*) \text{ and } (1-\eps)\sqrt{|S^*|/|T^*|}\leq z \leq \sqrt{|S^*|/|T^*|}.\]
With these approximate $D$ and $z$, and $\alpha = 1/(1 - \eps)$, we can follow the argument in \cref{theorem:peel-sample-approx} to see that a subgraph satisfying \condition exists, specifically $(S^*, T^*)$, and it has density at least $(1-2\eps)D$ in $G$ with probability at least $1 - n^{-5}$ as well.
Since $\Astream$ picks the subgraph corresponding to the largest $D$ that satisfies \condition, it finds a $(1-4\eps)$-approximate DDS. This is true with probability at least $1 - n^{-4}$ by taking the union bound over all instances.
\end{proof}

\subsection{Single-pass algorithm}
To remove the first pass of our two-pass algorithm, we need to simulate peeling at the same time as sampling.
However, this is difficult since it is possible that we are unable to determine which vertices should be removed during peeling until going through a majority of the stream.
By then, our sampling could potentially have saved too many edges and surpass $\omega(n\poly \log n)$ edges.

To circumvent this, our algorithm simply removes the first $k_T - 1$ edges of each vertex in $T$, if $z \geq 1$, or the first $k_S - 1$ edges of each vertex in $S$, if $z < 1$.
We observe that the set of edges removed in this way contains the set of edges that should be removed during peeling, and some additional edges.
However, we can no longer argue that all the edges of the densest subgraph are maintained, like in \PeelSample.
Though edges in the densest subgraph may be removed, we instead argue that not too many of its edges are removed and it still remains as an approximate densest subgraph.
We do this by peeling using thresholds $k_S = \eps D/z$ and $k_T = \eps Dz$ instead of $D/(2z)$ and $Dz/2$, respectively.
Leveraging this, we show that the additional edges that we lose do not affect the approximation much and we still compute a $(1-\eps)$-approximate DDS.

\begin{algorithm}
\caption{Single-pass semi-streaming algorithm for $(1-\eps)$-approximate DDS}
\label{alg:singlepass}
\textbf{Input:} vertex sets $(S, T)$, parameters $\eps$, $D$ and $z$
\begin{algorithmic}[1]
\State $k_S \gets \eps D/z, k_T \gets \eps Dz$
\State $d(v)\gets 0$ for all $v \in V$
\State $H\gets \emptyset$, $\alpha \gets 1/(1-\eps)$
\State $p \gets \min\rb{1, \frac{60\alpha\log n}{\eps^3 D}\max\rb{z, 1/z}}$
\While{stream is not empty}
    \State $(u, v) \gets $ next edge from the stream
    \If{$z \geq 1$}
        \State $d(v) = d(v) + 1$
        \If{$d(v) \geq k_T$}
            add $(u, v)$ to $H$ with prob. $p$
        \EndIf
    \Else
        \State $d(u) = d(u) + 1$
        \If{$d(u) \geq k_S$}
            add $(u, v)$ to $H$ with prob. $p$
        \EndIf
    \EndIf
    \If{$|H| > \frac{60(1+\eps)\alpha n\log n}{\eps^3}$}\label{line:singlepass:sample-size}
        $H\gets\emptyset$
    \EndIf
\EndWhile
\If{$z\geq 1$}
    remove all vertices $v\in T$ with $d(v) < k_T$
\Else{}
    remove all vertices $v\in S$ with $d(v) < k_S$
\EndIf
\If{$\rho(S,T)\geq D$\label{line:singlepass:den-check}}
    \Return $(S, T)$
\EndIf
\State $(S', T')\gets$ densest subgraph of $H$ that satisfies \condition, or $(\emptyset, \emptyset)$ if no subgraph exists
\State \Return $(S', T')$
\end{algorithmic}
\end{algorithm}
On line $13$ of \cref{alg:singlepass}, we make sure that if the number of edges exceeds $\frac{60(1+\eps)\alpha n\log n}{\eps^3}$, we set the edge set to $\emptyset$.
This is because if the number of edges goes above $\frac{60(1+\eps)\alpha n\log n}{\eps^3}$, with high probability the condition on line $16$ would be true and return a subgraph with density at least $D$ because of \cref{lemma:peel-sample-memory}.
We now show the following guarantee about \cref{alg:singlepass}.

\begin{lemma}
\label{lemma:singlepass:approx}
\cref{alg:singlepass} returns either a subgraph with density at least $(1-2\eps)D$ in $G$ or an empty subgraph with probability at least $1 - n^{-5}$. When given $(1-2\eps)\rho_G(S^*, T^*)\leq D \leq (1-\eps)\rho_G(S^*, T^*)$ and $(1-\eps)\sqrt{|S^*|/|T^*|} \leq z \leq \sqrt{|S^*|/|T^*|}$, \cref{alg:singlepass} returns a subgraph with density at least $(1-2\eps)D$ in $G$ with probability at least $1 - n^{-5}$.
\end{lemma}
\begin{proof}
If \cref{alg:singlepass} returns a subgraph that is not empty, it must satisfy \condition and therefore by \cref{lemma:sample-condition}, the subgraph has density at least $(1-2\eps)D$ in $G$ with probability at least $1 - n^{-5}$.

Let $D^* = \rho_G(S^*, T^*)$ and $z^* = \sqrt{|S^*|/|T^*|}$. Now, when given
\[(1-2\eps)D^*\leq D \leq (1-\eps)D^*\text{ and } (1-\eps)z^*\leq z \leq z^*,\]
we want to show that there exists a subgraph that satisfies \condition.
We will focus on the subgraph $G^* = (S^*, T^*, E_G(S^*, T^*))$.

If $z \geq 1$, then \cref{alg:singlepass} will remove $k_T - 1$ edges from each vertex in $T^*$.
Given the bounds on $D$ and $z$, we know that $k_T \leq \eps D^* z^*$.
Therefore, the number of edges remaining in the subgraph $(S^*, T^*)$ in $G$ is at least
\[|E_{G}(S^*, T^*)| - \eps D^* z^*|T^*| = (1-\eps)|E_{G}(S^*, T^*)|\]
since $D^* z^* |T^*| = |E_{G}(S^*, T^*)|$.
Because $\alpha = 1/(1-\eps)$ and the bounds on $D$ and $z$, we have that $(S^*, T^*)$ in the sampled graph $H$ will satisfy \condition with probability at least $1 - n^{-6}$ using the Chernoff bound.

If $z < 1$, then \cref{alg:singlepass} will remove $k_S - 1$ edges from each vertex in $S^*$.
Given the bounds on $D$ and $z$, we know that $k_S \leq \frac{\eps D^*}{z^*}$.
Therefore, the number of edges remaining in the subgraph $(S^*, T^*)$ in $G$ is at least
\[|E_{G}(S^*, T^*)| - \frac{\eps D^*}{z^*}|S^*| = (1-\eps)|E_{G}(S^*, T^*)|\]
since $D^*|S^*|/z^* = |E_{G}(S^*, T^*)|$.
Because $\alpha = 1/(1-\eps)$ and the bounds on $D$ and $z$, we have that $(S^*, T^*)$ in the sampled graph $H$ will satisfy \condition with probability at least $1 - n^{-6}$ using the Chernoff bound.

Therefore, there exists a subgraph that satisfies \condition and following \cref{theorem:sample-approx}, the subgraph returned will have density at least $(1-2\eps)D$ in $G$ with probability at least $1 - n^{-5}$.
\end{proof}

Using \cref{lemma:singlepass:approx}, we are now ready to present the full single-pass semi-streaming algorithm in \cref{theorem:singlepass}.

\begin{theorem}
\label{theorem:singlepass}
For $\eps \in (0, 1/4)$, there exists a single-pass semi-streaming algorithm using $\tO(n/\eps^3)$ memory that computes a $(1-4\eps)$-approximate DDS with probability at least $1 - n^{-4}$.
\end{theorem}
\begin{proof}
Let $\Astream$ be our single-pass semi-streaming algorithm. We now describe it and provide its analysis.

\paragraph{Algorithm description.}
$\Astream$ will run \cref{alg:singlepass} on all $m/n \leq D\leq n$ and $\frac{1}{\sqrt{n}} \leq z \leq \sqrt{n}$ in parallel using powers of $(1+\eps)$.
Then, out of all the outputs of \cref{alg:singlepass} where the vertex sets are both non-empty, the one corresponding to the largest $D$ is the final output of $\Astream$.

\paragraph{Algorithm memory.}
The size of $H$ never goes above $O(n\log n/\eps^3)$ edges due to line $13$ of \cref{alg:singlepass}.
Therefore, over all $O(\log^2 n)$ parallel instances of \cref{alg:singlepass}, $\Astream$ uses $O(n\log^3 n/\eps^3)$ memory in total.

\paragraph{Algorithm approximation.}
For all parallel instances of $D$ and $z$, if \cref{alg:singlepass} returns early on line $16$, a subgraph with density at least $D$ in $G$ is computed.
Otherwise, we observe that $H$ will never go above $\frac{60(1+\eps)\alpha n\log n}{\eps^3}$ with probability at least $1 - n^{-6}$ due to \cref{lemma:peel-sample-memory}.
So, $H$ contains a sample of the edges in $G$.
Additionally, we have the guarantee that the subgraph returned, if it is not empty, has density at least $(1-2\eps)D$ in $G$ with probability at least $1 - n^{-5}$ due to \cref{lemma:sample-condition}.
Now, we know there exists an instance where
\[(1-2\eps)\rho_G(S^*, T^*) \leq D \leq (1-\eps)\rho_G(S^*, T^*) \text{ and } (1-\eps)\sqrt{|S^*|/|T^*|}\leq z \leq \sqrt{|S^*|/|T^*|}.\]
Using \cref{lemma:singlepass:approx}, we have the guarantee that this specific instance of $D$ and $z$ will not return an empty subgraph and will return a subgraph with density at least $(1-2\eps)D$ in $G$ with probability at least $1 - n^{-5}$.
Since $\Astream$ picks the subgraph corresponding to the largest $D$ that satisfies \condition, it finds a $(1-4\eps)$-approximate DDS. This is true with probability at least $1 - n^{-4}$ by taking the union bound over all instances.
\end{proof}

Therefore, we have the first single-pass semi-streaming algorithm for DDS with a constant approximation.
Additionally, it completely closes the gap between undirected and directed DS by attaining a $(1-\eps)$-approximation.

\section{Near-linear MPC algorithm}
\label{section:mpc}
In this section, we show how to simulate \PeelSample directly in the near-linear MPC model.
Peeling can be simulated through computing the degrees of all vertices.
This can be calculated in $O(1)$ rounds using standard MPC primitives \cite{goodrich2011sorting}.
Then, when sampling from the remaining edges of the graph, the sample has $\tilde{O}(n)$ edges, which fits on a single machine.
This enables us to compute the densest subgraph of the sampled graph locally on a single machine.
Doing this for all instances of $D$ and $z$ in parallel requires only an extra $O(\poly\log n)$ factor in memory.
\begin{theorem}
\label{theorem:near-linear-MPC}
For $\eps\in(0, 1/4)$, there exists a MPC algorithm in the near-linear memory regime using $O(1)$ rounds that computes a $(1-4\eps)$-approximate DDS with probability at least $1 - n^{-4}$. The algorithm uses $\tilde{O}(n)$ memory per machine and $\tilde{O}(m)$ total memory.
\end{theorem}
\begin{proof}
Let $\AMPC$ be our near-linear MPC algorithm. We now describe it and provide its analysis.

\paragraph{Algorithm description.}
$\AMPC$ runs parallel instances on all $m/n \leq D\leq n$ and $\frac{1}{\sqrt{n}} \leq z \leq \sqrt{n}$ in parallel using powers of $(1+\eps)$. For a specific instance given $D$ and $z$, we do the following:
\begin{enumerate}
    \item First, we compute the degree of all vertices. If $z \geq 1$, we remove all vertices in $T$ with degree less than $k_T = Dz/2$. Otherwise, we remove all vertices in $S$ with degree less than $k_S = D/(2z)$.
    \item If the remaining induced subgraph has density at least $D$, we return it.
    \item We then uniformly sample the edges of the remaining induced subgraph using the probability $p = \min\rb{1, \frac{60\alpha\log n}{\eps^2 D}\max\rb{z, 1/z}}$ where $\alpha = 1/(1-\eps)$. Let us call the sampled subgraph $H$.
    \item Finally, we move all the edges of $H$ to a single machine and find the densest subgraph of $H$ that satisfies \condition, returning the subgraph or that no such subgraph exists.
\end{enumerate}
Then, out of all the nonempty vertex sets returned from each instance of $D$ and $z$, the one corresponding to the largest $D$ is the final output of $\AMPC$.

\paragraph{Algorithm memory and round complexity.}
The total memory is $\tilde{O}(m)$ since we make a copy of each edge for each instance of $D$ and $z$ so that they can be ran in parallel.
From \cref{lemma:peel-sample-memory}, the sample $H$ has size $O(n\log n)$ with probability at least $1 - n^{-6}$ which fits on a single machine with $\tO(n)$ memory.
As for the round complexity, calculating the degrees of all vertices uses $O(1)$ rounds.
Sampling and moving the edges to a single machine also takes $O(1)$ rounds.
Therefore, we use $O(1)$ rounds in total.

\paragraph{Algorithm approximation.}
$\AMPC$ directly simulates \cref{alg:twopass} for each instance of $D$ and $z$.
Therefore, following the argument in \cref{theorem:twopass}, $\AMPC$ finds a $(1-4\eps)$-approximate DDS with probability at least $1 - n^{-4}$.

\end{proof}

\section{Sublinear-time algorithm}
\label{section:sublineartime}
For our sublinear-time algorithm, the main idea is to once again simulate our \PeelSample algorithm using the oracle queries for given $D$ and $z$.
However, there are two small implementation details that we need to be careful of.

First, how does one attain a uniform sample of the edges remaining after peeling?
Generally when vertices are removed from a graph, it can be difficult with only the incident-list oracle to determine which edges are removed.
However, we observe that since our peeling is one-sided, removing vertices in $S$ or $T$ and never both, we can sample edges based on the vertex set that peeling is performed on.
For example, if $z \geq 1$ and peeling is performed on vertex set $T$, then we can sample based off the in-neighbors of each vertex in $T$ that is not removed during peeling.
On the other hand, if $z < 1$ and peeling is performed on vertex set $S$, then we can sample based off the out-neighbors of each vertex in $S$ that is not removed during peeling.

Second, though we may be able to use only $\tO(n)$ queries and $\tO(n)$ memory through simulating \PeelSample, we also need to make sure the algorithm can compute an approximate DDS in sublinear time as well.
Unlike the previous models with large local computation that allows them to compute the exact densest subgraph that satisfies \condition, we settle for an approximation.
Specifically, we use the LP algorithm from \cite{bahmani2014efficient} that returns a subgraph that satisfies \condition and takes near-linear time in the size of the input graph.

\subsection{LP algorithm satisfies \condition}
\label{section:algos-condition}
We briefly present the LP algorithm from \cite{bahmani2014efficient}.
In their work, they focus on solving the feasibility of the following dual LP for the DDS of graph $G$, $\Dual_G(D, z)$:
\begin{alignat*}{2}
\alpha_{eu} + \alpha_{ev} & \geq 1 \quad & \forall e &=\{u, v\} \in E_G\\
\sum_{(v, w)\in E_G} \alpha_{ev} & \leq D/(2z) \quad & \forall v &\in S\\
\sum_{(u, v)\in E_G} \alpha_{ev} & \leq Dz/2 \quad & \forall v &\in T\\
\alpha_{ev} &\leq 2\quad & \forall e &, v\\
\alpha_{ev} & \geq 0 \quad & \forall e &, v \\
\end{alignat*}
They have an additional constraint $\alpha_{ev} \leq 2$ which reduces the width of the LP to a constant.
One can observe that the smallest $D$ where $\Dual_G(D, z)$ is feasible is equal to the density of the DDS.
To solve the LP efficiently, they use the multiplicative weights update (MWU) method to compute an approximate feasible solution.
Specifically, they hope to decide the feasibility of the first line of constraints, $\alpha_{eu} + \alpha_{ev} \geq 1$, subject to the last four constraints, which they call $P_G(D,z)$ because they form polyhedral constraints.
Then, the MWU method iteratively assigns a vector of weights on the set of constraints $\alpha_{eu} + \alpha_{ev} \geq 1$ and approximately satisfies it.
In the end, the MWU method returns an assignment $\alpha$ that satisfies the constraints of $P_G(D,z)$ and also $\alpha_{eu} + \alpha_{ev}\geq 1-\eps$ for all $e=\{u, v\}\in E_G$.
It also returns a final vector of weights $y$, and they define the function
\[C_G(y, D, z) = \max_{\alpha \in P_G(D,z)} \sum_{v\in V}\sum_{e\text{ incident on } v} y_e\alpha_{ev}.\]
Finally, using a rounding scheme on the resulting weights $y$ and leveraging the function $C_G(y, D, z)$, they are able to recover a $(1-\eps)$-approximate DDS.
Since we want to solve the LP on the sampled graph $H$ returned by \Sample, we demonstrate the connection in feasibility between $\Dual_G$ and $\Dual_H$ in the following lemma.

\begin{lemma}
\label{lemma:lp-sample-feasibility}
Consider graph $G = (S, T, E)$, $\eps\in(0,1/2)$, $\alpha \geq 1$, $D \geq m/n$, and $z\in[1/\sqrt{n}, \sqrt{n}]$.
Let $H = \Sample(G, \eps, \alpha, D, z)$ be the sampled graph from $G$ with probability $p = \min\rb{1, \frac{60\alpha\log n}{\eps^2 D}\max\rb{z, 1/z}}$.
Then, $\Dual_H((1+\eps)pD^*, z^*)$ is feasible with probability at least $1 - 2n^{-4}$ where $D^* = \rho_G(S^*, T^*)$ and $z^* = \sqrt{|S^*|/|T^*|}$.
\end{lemma}
\begin{proof}
We know that $\Dual_G(D^*, z^*)$ is feasible.
Let $\alpha_{ev}$ be the solution that makes $\Dual_G(D^*, z^*)$ feasible.
We now show that $\alpha_{ev}$, for $e \in E_H$, makes $\Dual_H((1+\eps)pD^*, z^*)$ feasible with high probability.
For every $v\in S$,
\[\E{\sum_{(v, w)\in E_H} \alpha_{ev}} = p\sum_{(v, w)\in E_G} \alpha_{ev} \leq \frac{pD}{2z}.\]
Therefore,
\[\prob{\rb{\sum_{(v, w)\in E_H} \alpha_{ev}} - \frac{pD}{2z} \geq \frac{\eps pD}{2z}} \leq \exp\rb{-\eps^2 pD/(12z)} \leq n^{-5}\]
using the Chernoff bound and $\alpha_{ev} \leq 2$.
Similarly, for every $v\in T$,
\[\E{\sum_{(u, v)\in E_H} \alpha_{ev}} = p\sum_{(u, v)\in E_G} \alpha_{ev} \leq \frac{pDz}{2}.\]
Therefore,
\[\prob{\rb{\sum_{(u, v)\in E_H} \alpha_{ev}} - \frac{pDz}{2} \geq \frac{\eps pDz}{2}} \leq \exp\rb{-\eps^2 pDz/12} \leq n^{-5}\]
using the Chernoff bound  and $\alpha_{ev} \leq 2$.
Then, taking the union bound over all $v\in S\cup T$, we see that $\Dual_H((1+\eps)pD^*, z^*)$ is feasible with probability at least $1 - 2n^{-4}$.
\end{proof}

Now, we explicitly write the algorithm from \cite{bahmani2014efficient} to better parametrize $D$ and $z$.
We present the algorithm in \cref{alg:LP}.

\begin{algorithm}
\caption{LP algorithm for $(1-\eps)$-approximate DDS}
\label{alg:LP}
\textbf{Input:} bipartite graph $G = (S, T, E)$, parameters $\eps$, $D$ and $z$
\begin{algorithmic}[1]
\State $(\alpha, y)\gets $ run MWU method on $\Dual_G(D, z)$
\If{$\Dual_G(D, z)$ is infeasible or $\sum_e y_e < (1-6\eps)C_G(y, D, z)$}
\State \Return $(\emptyset, \emptyset)$, no subgraph exists
\EndIf
\State $(S', T')\gets$ approximate DS from applying rounding scheme to $G$, $y$, $D$, $z$
\State \Return $(S', T')$
\end{algorithmic}
\end{algorithm}

Then, we have the following properties of \cref{alg:LP}.

\begin{theorem}[\cite{bahmani2014efficient}, Theorem $7$ modified]
\label{theorem:lp-approx}
Consider graph $G = (S, T, E)$, $\eps\in(0,1/12)$, $D \geq m/n$, and $z\in[1/\sqrt{n}, \sqrt{n}]$. \cref{alg:LP} takes $O\rb{m\frac{\log^2 m}{\eps^3}}$ time and if it returns a nonempty subgraph $(S', T')$, it satisfies
\[|E_G(S', T')| \geq (1-9\eps)\frac{D}{2}\rb{|S'|/z + |T'|z}.\]
\end{theorem}
The proof of this theorem can be found in \cite{bahmani2014efficient}.
It is not hard to see that $D$ and $z$ can be approximated as long as $\Dual_G(D, z)$ remains feasible and the condition $\sum_e y_e < (1-6\eps)C_G(y, D, z)$ remains true after applying the MWU method.
Also, when \cref{alg:LP} outputs a nonempty subgraph, we observe that it satisfies \condition with $\alpha \geq 1/(1-9\eps)$ from \cref{theorem:lp-approx}.

\subsection{Full sublinear-time algorithm}

Using \cref{alg:LP}, we now have all the tools to make sure our algorithm takes sublinear-time. We present the full sublinear-time algorithm in \cref{theorem:sublineartime}.
\begin{theorem}
\label{theorem:sublineartime}
For $\eps\in(0, 1/12)$, there exists a sublinear-time algorithm that takes $\tilde{O}(n)$ time, $\tilde{O}(n)$ space, and uses $\tilde{O}(n)$ queries to an incident-list oracle model to compute a $(1-12\eps)$-approximate DDS with probability at least $1 - n^{-3}$.
\end{theorem}
\begin{proof}
Let $\Asub$ be our sublinear-time algorithm. We now describe it and provide its analysis.
\paragraph{Algorithm description.}
$\Asub$ runs instances on all $m/n \leq D\leq n$ and $\frac{1}{\sqrt{n}} \leq z \leq \sqrt{n}$ in parallel using powers of $(1+\eps)$. For a specific instance given $D$ and $z$, we do the following:
\begin{enumerate}
    \item First, we query the degrees of all vertices. If $z \geq 1$, we remove all vertices in $T$ with degree less than $Dz/2$. Otherwise, we remove all vertices in $S$ with degree less than $D/(2z)$.
    \item If the remaining induced subgraph has density at least $D$, we return it.
    \item Let $p = \min\rb{1, \frac{60\alpha\log n}{\eps^2 D}\max\rb{z, 1/z}}$ where $\alpha = 1/(1-9\eps)$ be our sampling probability.
    If $z \geq 1$, we sample edges by querying the in-neighbors of each $v\in T$ with probability $p$.
    Otherwise, we sample edges by querying the out-neighbors of each $v\in S$ with probability $p$.
    Let $H$ be the resulting sampled subgraph.
    \item Finally, we return the vertex sets computed by running \cref{alg:LP} on $H$ with $\eps$, $\frac{1+\eps}{1-3\eps}pD$, and $z$ as parameters.
\end{enumerate}
Then, out of all the nonempty vertex sets returned from each instance of $D$ and $z$, the one corresponding to the largest $D$ is the final output of $\Asub$.

\paragraph{Algorithm memory and query complexity.}
For a specific instance of $D$ and $z$, we know by \cref{lemma:peel-sample-memory} that the amount of memory stored is $O(n\log n)$ with probability at least $1 - n^{-6}$.
Similarly, the query complexity is $O(n\log n)$ where we queried degrees of vertices and the edges we wanted to sample.
Therefore, over all instances, we will use $O(n\log^3 n)$ memory and $O(n\log^3 n)$ queries.

\paragraph{Algorithm time complexity.}
For a specific instance of $D$ and $z$, we are running \cref{alg:LP} on a sampled graph with $O(n\log n)$ edges.
By \cref{theorem:lp-approx}, this will take $O(n\log^3 n)$ time.
Therefore, over all instances, we will take $O(n\log^5 n)$ time.

\paragraph{Algorithm approximation.}
For all instances of $D$ and $z$, if the algorithm returns early, a subgraph with density at least $D$ in $G$ is computed.
Otherwise, if running \cref{alg:LP} on the sampled subgraph $H$ outputs a nonempty subgraph, it satisfies \condition with $\alpha = 1/(1-9\eps)$ and the subgraph has density at least $(1-10\eps)D$ in $G$ with probability at least $1 - n^{-5}$ due to \cref{lemma:sample-condition}.
Now, we know there exists an instance where
\[(1-2\eps)\rho_G(S^*, T^*) \leq D \leq (1-\eps)\rho_G(S^*, T^*) \text{ and } (1-\eps)\sqrt{|S^*|/|T^*|}\leq z \leq \sqrt{|S^*|/|T^*|}.\]
With these parameters, we know that none of the edges in the densest subgraph $(S^*, T^*)$ will be removed during peeling, following the argument in \cref{theorem:peel-sample-approx}.
Let us consider running \cref{alg:LP} on $H$ with $\eps$, $\frac{1+\eps}{1-3\eps}pD$, and $z$ as parameters.
From \cref{lemma:lp-sample-feasibility}, we know that $\Dual_H\rb{\frac{1+\eps}{1-3\eps}pD, z}$ is feasible with probability at least $1 - 2n^{-4}$.
Additionally, from \cref{lemma:sample-DS}, $\rho_H(S^*, T^*) \geq (1-\eps)pD$ with probability at least $1 - n^{-6}$.
Therefore, if we let $D_H^*$ be the density of a DDS of $H$, we see that our density parameter of $\frac{1+\eps}{1-3\eps}pD$ is upper bounded by
\[\frac{1+\eps}{1-3\eps}pD \leq \frac{1+\eps}{1-4\eps}D_H^*\]
which implies that
\[\sum_e y_e \geq \frac{1-4\eps}{1+\eps}\cdot(1-\eps) C_H\rb{y, \frac{1+\eps}{1-3\eps}pD, z} \geq (1-6\eps)C_H\rb{y, \frac{1+\eps}{1-3\eps}pD, z}.\]
So, \cref{alg:LP} must return a nonempty subgraph and it will have density in $G$ at least
\[(1-10\eps)\cdot \frac{1+\eps}{1-3\eps} \cdot (1-2\eps)\rho_G(S^*, T^*) \geq (1-12\eps)\rho_G(S^*, T^*).\]
from \cref{theorem:lp-approx} and \cref{theorem:sample-approx} with probability at least $1 - n^{-5}$.
Overall, $\Asub$ computes our desired $(1-12\eps)$-approximation with probability at least $1 - n^{-3}$.

\end{proof}

\section*{Acknowledgement}
We disclose that we used ChatGPT-5.5 to analyze tradeoffs with the parameter $c$ in \cref{theorem:peel-sample-approx}, and its effect on the approximation guarantee and memory requirement in the single-pass semi-streaming algorithm.
This helped us observe that a $(1-\eps)$-approximate DDS single-pass semi-streaming algorithm would incur an extra factor of $\eps$ in the memory.

\bibliographystyle{alpha}
\bibliography{ref.bib}

@inproceedings{esfandiari2016brief,
  title={Brief announcement: Applications of uniform sampling: Densest subgraph and beyond},
  author={Esfandiari, Hossein and Hajiaghayi, MohammadTaghi and Woodruff, David P},
  booktitle={Proceedings of the 28th ACM Symposium on Parallelism in Algorithms and Architectures},
  pages={397--399},
  year={2016}
}

@InProceedings{pmlr-v235-mitrovic24a,
  title = 	 {Faster Streaming and Scalable Algorithms for Finding Directed Dense Subgraphs in Large Graphs},
  author =       {Mitrovic, Slobodan and Pan, Theodore},
  booktitle = 	 {Proceedings of the 41st International Conference on Machine Learning},
  pages = 	 {35876--35891},
  year = 	 {2024},
  volume = 	 {235},
  series = 	 {Proceedings of Machine Learning Research},
  publisher =    {PMLR}
}

@inproceedings{ghaffari2025density,
  title={Density-Dependent Graph Orientation and Coloring in Scalable MPC},
  author={Ghaffari, Mohsen and Grunau, Christoph},
  booktitle={Proceedings of the ACM Symposium on Principles of Distributed Computing},
  pages={349--359},
  year={2025}
}

@inproceedings{bhattacharya2015space,
  title={Space-and time-efficient algorithm for maintaining dense subgraphs on one-pass dynamic streams},
  author={Bhattacharya, Sayan and Henzinger, Monika and Nanongkai, Danupon and Tsourakakis, Charalampos},
  booktitle={Proceedings of the forty-seventh annual ACM symposium on Theory of computing},
  pages={173--182},
  year={2015}
}

@inproceedings{
mitrovic2025new,
title={New Parallel and Streaming Algorithms for Directed Densest Subgraph},
author={Slobodan Mitrovi{\'c} and Theodore Pan and Mahdi Qaempanah and Mohammad Amin Raeisi},
booktitle={The Thirty-ninth Annual Conference on Neural Information Processing Systems},
year={2025}
}

@inproceedings{bahmani2014efficient,
  title={Efficient primal-dual graph algorithms for mapreduce},
  author={Bahmani, Bahman and Goel, Ashish and Munagala, Kamesh},
  booktitle={International Workshop on Algorithms and Models for the Web-Graph},
  pages={59--78},
  year={2014},
  organization={Springer}
}

@article{bahmani2012densest,
  title={Densest Subgraph in Streaming and MapReduce},
  author={Bahmani, Bahman and Kumar, Ravi and Vassilvitskii, Sergei},
  journal={Proceedings of the VLDB Endowment},
  volume={5},
  number={5},
  year={2012}
}

@inproceedings{charikar2000greedy,
  title={Greedy approximation algorithms for finding dense components in a graph},
  author={Charikar, Moses},
  booktitle={International workshop on approximation algorithms for combinatorial optimization},
  pages={84--95},
  year={2000},
  organization={Springer}
}

@inproceedings{ghaffari2019improved,
  title={Improved parallel algorithms for density-based network clustering},
  author={Ghaffari, Mohsen and Lattanzi, Silvio and Mitrovi{\'c}, Slobodan},
  booktitle={International Conference on Machine Learning},
  pages={2201--2210},
  year={2019},
  organization={PMLR}
}

@book{kannan1999analyzing,
  title={Analyzing the structure of large graphs},
  author={Kannan, Ravindran and Vinay, V},
  year={1999},
  publisher={Universit{\"a}t Bonn. Institut f{\"u}r {\"O}konometrie und Operations Research}
}

@inproceedings{su2020distributed,
  title={Distributed Dense Subgraph Detection and Low Outdegree Orientation},
  author={Su, Hsin-Hao and Vu, Hoa T},
  booktitle={34th International Symposium on Distributed Computing},
  year={2020}
}

@inproceedings{sawlani2020near,
  title={Near-optimal fully dynamic densest subgraph},
  author={Sawlani, Saurabh and Wang, Junxing},
  booktitle={Proceedings of the 52nd Annual ACM SIGACT Symposium on Theory of Computing},
  pages={181--193},
  year={2020}
}

@inproceedings{dinitz2025almost,
  title={Almost tight bounds for differentially private densest subgraph},
  author={Dinitz, Michael and Kale, Satyen and Lattanzi, Silvio and Vassilvitskii, Sergei},
  booktitle={Proceedings of the 2025 Annual ACM-SIAM Symposium on Discrete Algorithms (SODA)},
  pages={2908--2950},
  year={2025},
  organization={SIAM}
}

@inproceedings{hooi2016fraudar,
  title={Fraudar: Bounding graph fraud in the face of camouflage},
  author={Hooi, Bryan and Song, Hyun Ah and Beutel, Alex and Shah, Neil and Shin, Kijung and Faloutsos, Christos},
  booktitle={Proceedings of the 22nd ACM SIGKDD international conference on knowledge discovery and data mining},
  pages={895--904},
  year={2016}
}

@inproceedings{prakash2010eigenspokes,
  title={Eigenspokes: Surprising patterns and scalable community chipping in large graphs},
  author={Prakash, B Aditya and Sridharan, Ashwin and Seshadri, Mukund and Machiraju, Sridhar and Faloutsos, Christos},
  booktitle={Pacific-Asia conference on knowledge discovery and data mining},
  pages={435--448},
  year={2010},
  organization={Springer}
}

@article{kleinberg1999authoritative,
  title={Authoritative sources in a hyperlinked environment},
  author={Kleinberg, Jon M},
  journal={Journal of the ACM (JACM)},
  volume={46},
  number={5},
  pages={604--632},
  year={1999},
  publisher={ACM New York, NY, USA}
}

@article{andersen2010local,
  title={A local algorithm for finding dense subgraphs},
  author={Andersen, Reid},
  journal={ACM Transactions on Algorithms (TALG)},
  volume={6},
  number={4},
  pages={1--12},
  year={2010},
  publisher={ACM New York, NY, USA}
}

@article{chen2023densest,
  title={Densest multipartite subgraph search in heterogeneous information networks},
  author={Chen, Lu and Liu, Chengfei and Zhou, Rui and Liao, Kewen and Xu, Jiajie and Li, Jianxin},
  journal={Proceedings of the VLDB Endowment},
  volume={17},
  number={4},
  pages={699--711},
  year={2023},
  publisher={VLDB Endowment}
}

@article{ma2021directed,
  title={On directed densest subgraph discovery},
  author={Ma, Chenhao and Fang, Yixiang and Cheng, Reynold and Lakshmanan, Laks VS and Zhang, Wenjie and Lin, Xuemin},
  journal={ACM Transactions on Database Systems (TODS)},
  volume={46},
  number={4},
  pages={1--45},
  year={2021},
  publisher={ACM New York, NY}
}

@inproceedings{saha2010dense,
  title={Dense subgraphs with restrictions and applications to gene annotation graphs},
  author={Saha, Barna and Hoch, Allison and Khuller, Samir and Raschid, Louiqa and Zhang, Xiao-Ning},
  booktitle={Annual International Conference on Research in Computational Molecular Biology},
  pages={456--472},
  year={2010},
  organization={Springer}
}

@book{junker2008analysis,
  title={Analysis of biological networks},
  author={Junker, Bj{\"o}rn H and Schreiber, Falk},
  volume={2},
  year={2008},
  publisher={Wiley Online Library}
}

@inproceedings{goodrich2011sorting,
  title={Sorting, searching, and simulation in the mapreduce framework},
  author={Goodrich, Michael T and Sitchinava, Nodari and Zhang, Qin},
  booktitle={International Symposium on Algorithms and Computation},
  pages={374--383},
  year={2011},
  organization={Springer}
}

@inproceedings{barenboim2008sublogarithmic,
  title={Sublogarithmic distributed MIS algorithm for sparse graphs using Nash-Williams decomposition},
  author={Barenboim, Leonid and Elkin, Michael},
  booktitle={Proceedings of the twenty-seventh ACM symposium on Principles of distributed computing},
  pages={25--34},
  year={2008}
}

@article{picard1982network,
  title={A network flow solution to some nonlinear 0-1 programming problems, with applications to graph theory},
  author={Picard, Jean-Claude and Queyranne, Maurice},
  journal={Networks},
  volume={12},
  number={2},
  pages={141--159},
  year={1982},
  publisher={Wiley Online Library}
}

@inproceedings{esfandiari2018parallel,
  title={Parallel and streaming algorithms for k-core decomposition},
  author={Esfandiari, Hossein and Lattanzi, Silvio and Mirrokni, Vahab},
  booktitle={International Conference on Machine Learning},
  pages={1397--1406},
  year={2018},
  organization={PMLR}
}

@article{matula1983smallest,
  title={Smallest-last ordering and clustering and graph coloring algorithms},
  author={Matula, David W and Beck, Leland L},
  journal={Journal of the ACM (JACM)},
  volume={30},
  number={3},
  pages={417--427},
  year={1983},
  publisher={ACM New York, NY, USA}
}

@InProceedings{10.1007/3-540-61422-2_127,
author="Asahiro, Yuichi
and Iwama, Kazuo
and Tamaki, Hisao
and Tokuyama, Takeshi",
title="Greedily finding a dense subgraph",
booktitle="Algorithm Theory --- SWAT'96",
year="1996",
publisher="Springer Berlin Heidelberg",
pages="136--148"
}

@inproceedings{khuller2009finding,
  title={On finding dense subgraphs},
  author={Khuller, Samir and Saha, Barna},
  booktitle={International colloquium on automata, languages, and programming},
  pages={597--608},
  year={2009},
  organization={Springer}
}


\end{document}